\newif\iflong
\longtrue

\documentclass[runningheads]{llncs}
\usepackage[T1]{fontenc}
\usepackage{graphicx}
\usepackage[hypertexnames=false]{hyperref}
\usepackage{color}

\usepackage{listings}
\lstdefinelanguage{program}{%
  keywords={%
    let,pass,function,%
    var,const,bool,int,void,atomic,%
    while,do,if,then,else,assume,assert,call,return,rule,forall,with,new,choose,skip,%
    task,async,yield,for,wait,%
    type,relation,init, action, safety, invariant, axiom, input,repeat
  },
  morecomment=[l]{//},
  morecomment=[s]{/*}{*/},
  morecomment=[n]{(*}{*)},
  mathescape=true,
  escapeinside=`',
}
\usepackage{booktabs}
\usepackage{mathpartir}
\usepackage{multirow}
\usepackage{amsmath}
\usepackage{amssymb}
\usepackage{stmaryrd}
\usepackage{array}
\usepackage{subcaption}
\usepackage[nodayofweek]{datetime}
\usepackage{graphicx}
\usepackage{algorithm}
\usepackage{wrapfig}
\usepackage[noend]{algpseudocode}
\usepackage[flushleft]{threeparttable}
\usepackage{stmaryrd}
\usepackage{multicol}
\usepackage{pifont}
\usepackage[xcolor]{changebar}
\usepackage{caption}
\cbcolor{magenta}
\usepackage{tcolorbox}
\tcbuselibrary{theorems}
\usepackage{centernot}
\usepackage{galois}
\usepackage{musicography}
\usepackage{scalerel}
\usepackage[numbers]{natbib}

\usepackage{xcolor}
\usepackage{extarrows}
\usepackage{mathtools}
\usepackage{paralist}

\usepackage{tikz}
\usetikzlibrary{positioning,fit,calc,arrows.meta,shapes.geometric,decorations.pathmorphing}
\makeatletter
\newcommand\currentcoordinate{\the\tikz@lastxsaved,\the\tikz@lastysaved}
\makeatother

\usepackage[noabbrev,capitalise]{cleveref} %

\iflong
\newcommand{\refappendix}[1]{\Cref{#1}}
\else
\newcommand{\refappendix}[1]{the extended version~\cite{extendedVersion}}
\fi

\iflong
\newcommand{\toolong}[1]{#1}
\else
\newcommand{\toolong}[1]{}
\fi

\crefformat{section}{\S#2#1#3} %
\crefformat{subsection}{\S#2#1#3}
\crefformat{subsubsection}{\S#2#1#3}

\Crefname{equation}{Eq.}{Eqs.}
\Crefname{conjecture}{Conjecture}{Conjectures}
\Crefname{proposition}{Proposition}{Propositions}
\Crefname{lemma}{Lemma}{Lemmas}
\Crefname{corollary}{Corollary}{Corollaries}
\Crefname{example}{Example}{Examples}
\Crefname{definition}{Def.}{Defs.}
\Crefname{algorithm}{Alg.}{Alg.}
\Crefname{theorem}{Thm.}{Thms.}
\Crefname{figure}{Fig.}{Fig.}
\crefname{line}{line}{lines}

\newcommand{\para}[1]{\vspace{2pt}\noindent\textbf{\textit{#1.}}}

\newcommand{\dom}[1]{dom({#1})}
\newcommand{\ov}{\overline}

\newcommand{\card}[1]{{\left\vert{#1}\right\vert}} %

\renewcommand{\implies}{\Longrightarrow}

\newcommand{\notimplies}{\centernot\implies}

\newcommand{\true}{{\textit{true}}}
\newcommand{\false}{{\textit{false}}}
\newcommand{\vocabulary}{\Sigma}
\newcommand{\voc}{\vocabulary}

\newcommand{\Init}{{\textit{Init}}}
\newcommand{\Bad}{\textit{Bad}}

\newcommand{\States}{{\mbox{States}}[\voc]}

\newcommand{\tr}{\delta}

\newcommand{\SATKw}{\textsc{sat}}
\newcommand{\SAT}[1]{\SATKw(#1)}

\newcommand{\Frameai}{\xi}

\newcommand{\NP}{\textbf{NP}}

\renewcommand{\vec}{\ov}

\newcommand{\project}[2]{\pi_{#2}({#1})}

\newcommand{\set}[1]{\{{#1}\}}

\newcommand{\eqdef}{\stackrel{\rm def}{=}}

\newcommand{\prop}{x}

\newcommand{\reflextr}[1]{\underline{#1}}

\newcommand{\bmcunroll}[2]{{\reflextr{#1}}^{{#2}}}
\newcommand{\bmc}[3]{\bmcunroll{#1}{#3}({#2})}

\algdef{SE}[DOWHILE]{Do}{doWhile}{\algorithmicdo}[1]{\algorithmicwhile\ #1}%

\newcommand{\dnfsize}[1]{\card{#1}_{\rm dnf}}
\newcommand{\cnfsize}[1]{\card{#1}_{\rm cnf}}
\newcommand{\cubemon}[2]{\textit{cube}_{{#2}}({#1})}
\newcommand{\moncube}[2]{\cubemon{#1}{#2}}
\newcommand{\monox}[2]{\mathcal{M}_{#2}({#1})}

\newcommand{\boundarypos}[1]{\partial^{+}({#1})}
\newcommand{\boundaryneg}[1]{\partial^{-}({#1})}
\newcommand{\bigO}{O}

\newcommand{\equivalencequery}[1]{{\rm EQ}\left({#1}\right)}
\newcommand{\membershipquery}[1]{{\rm MQ}\left({#1}\right)}

\newcommand{\mspan}[1]{{\rm MSpan}({#1})}

\newcommand{\mhull}[2]{{\rm MHull}_{#2}({#1})}

\newcommand{\absr}[1]{{#1}^{\musDoubleSharp}}
\newcommand{\malpha}[1]{\alpha_{#1}}
\newcommand{\madom}[1]{\mathbb{M}[{#1}]}

\newcommand{\bkcube}{b}

\newcommand{\reflect}[1]{\textit{Ref}({#1})}

\newcommand{\join}{\mathbin{\sqcup}}

\newcommand{\litabs}[1]{\textcolor{blue}{#1}}
\newcommand{\cubdom}[1]{\dom{#1}}

\makeatletter
\newcommand{\dotsym}[1]{{\vphantom{#1}\mathpalette\d@tsym{#1}\relax}}
\newcommand{\d@tsym}[1]{%
  \ooalign{\hidewidth$\m@th\cdot$\hidewidth\cr$\m@th #1$\cr}%
}
\makeatother
\makeatletter
\newcommand{\dotdelta}{{\vphantom{\delta}\mathpalette\d@td@lta\relax}}
\newcommand{\d@td@lta}[2]{%
  \ooalign{\hidewidth$\m@th#1\cdot$\hidewidth\cr$\m@th#1\delta$\cr}%
}
\makeatother

\newcommand\mathbox[1]{\mathord{\ThisStyle{%
  \fboxsep0\LMpt\relax\kern1\LMpt\fbox{$\SavedStyle#1$}\kern1\LMpt}}}
\newcommand{\cubejoin}[1]{\mathbox{#1}}

\newcommand{\restrict}[2]{{#1}\big|_{#2}}

\newcommand{\superefficientalgname}{Monotonize}
\newcommand{\dualourinterpolationalgname}{\mbox{\rm Dual-MB-ITP}}

\newcommand{\cubedom}[1]{\cubdom{#1}}

\newcommand{\aibound}{\zeta}

\newcommand{\hamminginterval}[2]{\cubejoin{{\, #1},{#2 \,}}}

\begin{document}

\newif\ifcomments
\commentsfalse
\nochangebars
\definecolor{dg}{cmyk}{0.60,0,0.88,0.27}

\newcommand{\sharonnew}[1]{\sharon{#1}}
\newcommand{\yotamnew}[1]{\yotamsmall{#1}}
\newcommand{\yotamforlater}[1]{}

\ifcomments
\newcommand{\artem}[1]{{\footnotesize\color{olive}[{\bf Artem}: #1]}}
\newcommand{\yotamsmall}[1]{{\footnotesize\color{magenta}[{\bf Yotam}: #1]}}

\newcommand{\sharon}[1]{{\textcolor{purple}{SS: {\em #1}}}}
\newcommand{\mooly}[1]{{\textcolor{cyan}{MS: {\em #1}}}}
\newcommand{\yotam}[1]{{\textcolor{magenta}{{\bf #1}}}}
\newcommand{\jrw}[1]{{\textcolor{green}{JRW: {\em #1}}}}
\newcommand{\TODO}[1]{{\textcolor{red}{TODO: {\em #1}}}}

\else
\newcommand{\sharon}[1]{}
\newcommand{\adam}[1]{}
\newcommand{\mooly}[1]{}
\newcommand{\neil}[1]{}
\newcommand{\jrw}[1]{}
\newcommand{\yotam}[1]{}
\newcommand{\TODO}[1]{}
\newcommand{\artem}[1]{}
\newcommand{\yotamsmall}[1]{}

\fi

\newcommand{\commentout}[1]{}
\newcommand{\OMIT}[1]{}  %
\title{Invariant Inference With Provable Complexity From the Monotone Theory}
\author{Yotam M.\ Y.\ Feldman \and Sharon Shoham}
\institute{Tel Aviv University}
\maketitle              %
\begin{abstract}
\vspace{-0.5cm}
Invariant inference algorithms such as interpolation-based inference and IC3/PDR show that it is feasible, in practice, to find inductive invariants for many interesting systems, but non-trivial upper bounds on the computational complexity of such algorithms are scarce, and limited to simple syntactic forms of invariants.
In this paper we achieve invariant inference algorithms, in the domain of propositional transition systems, with \emph{provable upper bounds} on the number of SAT calls.
We do this by building on the \emph{monotone theory}, developed by Bshouty for exact learning Boolean formulas.
We prove results for two invariant inference frameworks:
(i) \emph{model-based interpolation}, where we show an algorithm that, under certain conditions about reachability, %
efficiently infers invariants
when they have both short CNF and DNF representations (transcending previous results about monotone invariants); and
(ii) \emph{abstract interpretation} in a domain based on the monotone theory that was previously studied in relation to \emph{property-directed reachability}, where we propose an efficient implementation of the best abstract transformer, leading to overall complexity bounds on the number of SAT calls.
These results build on a novel procedure for computing least monotone overapproximations.
\vspace{-0.3cm}
\end{abstract} 
\section{Introduction}
In a fruitful, recent trend, many that aspire to innovate in verification seek advice from research in machine learning~\cite[e.g.][]{DBLP:conf/cav/SharmaNA12,DBLP:conf/sas/0001GHAN13,DBLP:conf/esop/0001GHALN13,DBLP:conf/icse/JhaGST10,ICELearning,DBLP:journals/pacmpl/EzudheenND0M18,DBLP:journals/acta/JhaS17,DBLP:journals/pacmpl/FeldmanISS20,DBLP:conf/pldi/KoenigPIA20,DBLP:journals/pacmpl/FeldmanSSW21}.
The topic of this paper is the application of the \emph{monotone theory}, developed by Bshouty in exact concept learning, to improve theoretical complexity results for \emph{inductive invariant inference}.

One of the \emph{modi operandi} for automatically proving that a system is safe---that it cannot reach a state it should not---is an \emph{inductive invariant}, which is an assertion that
\begin{inparaenum}[(i)]
	\item holds for the initial states,
	\item does not hold in any bad state, and
	\item is closed under transitions of the system.
\end{inparaenum}
These properties are reminiscent of a data classifier, separating good from bad points, prompting the adaptation of algorithms from classical classification to invariant inference~\cite[e.g.][]{DBLP:conf/cav/SharmaNA12,DBLP:conf/sas/0001GHAN13,DBLP:conf/esop/0001GHALN13,DBLP:journals/fmsd/SharmaA16,ICELearning,DBLP:journals/pacmpl/EzudheenND0M18,DBLP:conf/popl/0001NMR16,DBLP:journals/pacmpl/FeldmanSSW21}. In this paper we focus on inductive invariants for propositional transition systems, which are customary in hardware verification and also applicable to software systems through predicate abstraction~\cite{DBLP:conf/cav/GrafS97,DBLP:conf/popl/FlanaganQ02}.

The \emph{monotone theory} %
by Bshouty~\cite{DBLP:journals/iandc/Bshouty95} is a celebrated achievement in learning theory (most notably in exact learning with queries~\cite{DBLP:journals/ml/Angluin87}) that is the foundation for learning Boolean formulas with complex syntactic structures. At its core, the monotone theory studies the
\emph{monotonization} $\monox{\varphi}{\bkcube}$ of a formula $\varphi$ w.r.t.\ a valuation $\bkcube$, which is the smallest $\bkcube$-monotone formula that overapproximates $\varphi$. (This concept is explained in~\Cref{sec:monotone-background}.) %
In Bshouty's work, several monotonizations are used to efficiently reconstruct $\varphi$.

Recently, the monotone theory has been applied to theoretical studies of invariant inference in the context of two prominent SAT-based inference approaches. In this paper, we solve open problems in each, using a new efficient algorithm to compute monotonizations.

\para{Efficient interpolation-based inference}
The study~\cite{DBLP:journals/pacmpl/FeldmanSSW21} of interpolation-based invariant inference---a hugely influential approach pioneered by McMillan~\cite{DBLP:conf/cav/McMillan03}%
---identified the fence condition as a property of systems and invariants under which the success of a model-based inference algorithm~\cite{DBLP:conf/hvc/ChocklerIM12,DBLP:conf/lpar/BjornerGKL13} is guaranteed. %
(We explain the fence condition in~\Cref{sec:background-itp}.)
Under this condition, %
the number of SAT calls (specifically, bounded model checking queries)
of the original model-based interpolation algorithm was shown to be polynomial in the DNF size (the number of terms in the smallest DNF representation) of the invariant, but only when the invariant is \emph{monotone} (containing no negated variables)~\cite{DBLP:journals/pacmpl/FeldmanSSW21}. Based on the monotone theory, %
the authors of~\cite{DBLP:journals/pacmpl/FeldmanSSW21} further
introduced an algorithm that, under the same fence condition, could efficiently infer invariants that were \emph{almost} monotone (containing $\bigO(1)$ terms with negated variables).
However, their techniques could not extend to mimic the pinnacle result of Bshouty's paper: the CDNF algorithm~\cite{DBLP:journals/iandc/Bshouty95}, which can learn formulas in a number of queries that is
polynomial in their DNF size, their CNF size (the number of clauses in their smallest CNF representation), and the number of variables. It was unclear whether an analogous result is possible in invariant inference without strengthening the fence condition, %
e.g.\ to assume that the fence condition holds both forwards and backwards (see~\Cref{sec:related}).

We solve this question, and introduce an algorithm that can infer an invariant in a number of SAT queries (specifically, bounded model checking queries)
that is polynomial in the invariant's DNF size, CNF size, and the number of variables, under the assumption that the invariant satisfies the fence condition, without further restrictions (\Cref{thm:lambda-dual-itp-efficient}). In particular, this implies that invariants that are representable by a small decision tree can be inferred efficiently.
The basic idea is to learn an invariant $I$ as a conjunction of monotonizations $\monox{I}{\sigma}$ where $\sigma$ are chosen as counterexamples to induction. The challenge is that $I$ is unknown, and
the relatively weak assumption on the transition system of the fence condition does not allow the use of several operations (e.g.\ membership queries) that learning algorithms rely on to efficiently generate such a representation. %

\para{Efficient abstract interpretation}
The study of IC3/PDR~\cite{ic3,pdr} revealed that part of the overapproximation this sophisticated algorithm performs is captured by an abstract interpretation procedure, in an abstract domain founded on the monotone theory~\cite{DBLP:journals/pacmpl/FeldmanSSW22}. In this procedure, dubbed $\Lambda$-PDR, each iteration involves several monotonizations of the set of states reachable in one step from the value of the previous iteration. Upper bounds on the number of iterations in $\Lambda$-PDR were investigated to shed light on the number of frames of PDR~\cite{DBLP:journals/pacmpl/FeldmanSSW22}.
However, it was unclear whether the abstract domain itself can be implemented in an efficient manner, and whether efficient complexity bounds on the number of SAT queries (and not just the number of iterations) in $\Lambda$-PDR can be obtained.

We solve this question, and show that $\Lambda$-PDR can be implemented to yield an overall upper bound on the number of SAT calls which is polynomial in the same quantity that was previously used to bound the number of iterations in $\Lambda$-PDR (\Cref{thm:efficient-eepdr}).
This is surprising because, until now, %
there was no way to compute monotonizations of the post-image of the previous iteration that did not suffer from the fact that the exact post-image of a set of states may be much more complex to represent than its abstraction.

\para{Super-efficient monotonization}
Bshouty~\cite{DBLP:journals/iandc/Bshouty95} provided an algorithm to compute the monotonization $\monox{\varphi}{\bkcube}$,
but the complexity of this algorithm depends on the DNF size of the original formula $\varphi$.
Our aforementioned results build on a new algorithm for the same task, whose complexity depends on the DNF size of the monotonization $\monox{\varphi}{\bkcube}$ (\Cref{thm:efficient-monox-efficiency}), which may be much smaller (and never larger).
This enables our efficient interpolation-based inference algorithm and our efficient implementation of abstract interpretation, although each result requires additional technical sophistication:
For our efficient model-based interpolation result, the key idea is that the monotonization of an invariant satisfying the fence condition can be computed through the monotonization of the set of states reachable in at most $s$ steps, and our new monotonization algorithm allows to do this efficiently even when the latter set is complex to represent exactly.
For our efficient abstract interpretation result, the key idea is that the DNF size of an abstract iterate is bounded by a quantity related to monotonizations of the transition relation, and our new monotonization algorithm allows to compute it efficiently w.r.t\ the same quantity even though the DNF size of the exact post-image of the previous iterate may be larger.

Overall, we make the following contributions:
\begin{itemize}
	\item We introduce a new efficient algorithm to compute monotonizations, whose complexity in terms of the number of SAT queries is proportional to the DNF size of its output (\Cref{sec:efficient-monox}).

	\item We prove that an invariant that satisfies the fence condition can be inferred in a number of SAT (bounded model checking) queries that is polynomial in its CNF size, its DNF size, and the number of variables; in particular, invariants represented by small decision trees can be efficiently inferred (\Cref{sec:cdnf}).

	\item We prove an efficient complexity upper bound for the number of SAT queries performed by abstract interpretation in a domain based on the monotone theory (\Cref{sec:efficient-eedpr}).
\end{itemize}
\Cref{sec:prelim} sets preliminary notation and~\Cref{sec:monotone-background} provides background on the monotone theory. \Cref{sec:related} discusses related work and~\Cref{sec:conclusion} concludes.  
\section{Preliminaries}
\label{sec:prelim}
We work with propositional transition systems defined over a vocabulary $\voc = \set{p_1,\ldots,p_n}$ of $n$ Boolean variables.
We identify a formula with the set of its valuations, and at times also identify a set of valuations with an arbitrary formula that represents it which is chosen arbitrarily (one always exists in propositional logic). $\varphi \implies \psi$ denotes the validity of the formula $\varphi \to \psi$.

\para{States, transition systems, inductive invariants}
A \emph{state} is a \emph{valuation} to $\voc$. %
If $x$ is a state, $x[p]$ is the value ($\true/\false$ or $1/0$) that $x$ assigns to the variable $p \in \voc$.
A \emph{transition system} is a triple $(\Init,\tr,\Bad)$ where $\Init,\Bad$ are formulas over $\voc$ denoting the set of initial and bad states respectively, and the \emph{transition relation} $\tr$ is a formula over $\voc \uplus \voc'$, where $\voc' = \{ \prop' \mid \prop \in \voc\}$ is a copy of the vocabulary used to describe the post-state of a transition.
If $\tilde{\voc},\tilde{\voc}'$ are distinct copies of $\voc$, $\tr[\tilde{\voc},\tilde{\voc}']$ denotes the substitution in $\tr$ of each $p \in \voc$ by its corresponding in $\tilde{\voc}$ and likewise for $\voc',\tilde{\voc'}$.
Given a set of states $S$, the \emph{post-image} of $S$ is $\tr(S) = \set{\sigma' \mid \exists \sigma \in S. \ (\sigma,\sigma') \models \tr}$.
A transition system is \emph{safe} if all the states that are reachable from $\Init$ via any number of steps of $\tr$ satisfy $\neg \Bad$. %
An \emph{inductive invariant} is a formula $I$ over $\voc$ such that
\begin{inparaenum}[(i)]
	\item $\Init \implies I$,
	\item \label{it:prelim-inductiveness} $I \land \tr \implies I'$, and
	\item $I \implies \neg\Bad$, where $I'$ %
denotes the result of substituting each $\prop \in \voc$ for $\prop' \in \voc'$ in $I$.
\end{inparaenum}
In the context of propositional logic, a transition system is safe iff it has an inductive invariant.

\para{Use of SAT in invariant inference}
Given a candidate, the requirements for being an inductive invariant can be verified using SAT; we refer to the SAT query that checks requirement~\ref{it:prelim-inductiveness} by the name \emph{inductiveness check}.
When an inductiveness checks fails, a SAT solver returns a \emph{counterexample to induction}, which is a transition $(\sigma,\sigma')$ with $\sigma \models I$ but $\sigma' \not\models I$.
Another important check in invariant inference algorithms that can be implemented using SAT is \emph{bounded model checking (BMC)}~\cite{DBLP:conf/tacas/BiereCCZ99}, which asks whether a set of states described by a formula $\psi$ is forwards unreachable in a bounded number $s \in \mathbb{N}$ of steps; we write this as the check $\bmc{\tr}{\Init}{s} \cap \psi \overset{?}{=} \emptyset$. %
Using SAT it is also possible to obtain a counterexample $\sigma \in \bmc{\tr}{\Init}{s} \cap \psi$ if it exists.

We measure the \emph{complexity} of a SAT-based inference algorithm by the number of SAT calls it performs (including inductiveness checks, BMC, and other SAT calls), and the number of other steps, when each SAT call is considered one step (an oracle call).

\para{Literals, Cubes, Clauses, CNF, DNF}
A \emph{literal} $\ell$ is a variable $p$ or its negation $\neg p$.
A \emph{clause} $c$ is a disjunction of %
literals. %
The empty clause is $\false$.
A formula is in \emph{conjunctive normal norm (CNF)} if it is a conjunction of clauses.
A \emph{cube} or \emph{term} $d$ is a conjunction of a consistent set of literals; at times, we also refer directly to the set and write $\ell \in d$. The empty cube is $\true$.
A formula is in \emph{disjunctive normal form (DNF)} if it is a disjunction of terms.
The \emph{domain}, $\cubdom{d}$, of a cube $d$ is the set of variables that appear in it (positively or negatively).
Given a state $\sigma$, we use the state and the (full) cube that consists of all the literals that are satisfied in $\sigma$ interchangeably.
$\dnfsize{\varphi}$ is the minimal number of terms in any DNF representation of $\varphi$.
$\cnfsize{\varphi}$ is the minimal number of clauses in any CNF representation of $\varphi$.
\section{Background: The Monotone Theory}
\label{sec:monotone-background}
This section provides necessary definitions and results from the monotone theory by Bshouty~\cite{DBLP:journals/iandc/Bshouty95} as used in this paper. Our presentation is based on~\cite{DBLP:journals/pacmpl/FeldmanSSW21,DBLP:journals/pacmpl/FeldmanSSW22} (lemmas that are stated here slightly differently are proved in~\refappendix{sec:proofs-appendix}).

Boolean functions which are \emph{monotone} are special in many ways; one is that they are easier to learn~\cite[e.g.][]{DBLP:journals/cacm/Valiant84,DBLP:journals/ml/Angluin87}. Syntactically, a monotone function can be written in DNF so that all variables appear positively. This is easily generalized to $b$-monotone formulas, where each variable appears only at one polarity specified by $b$ (\Cref{def:b-monotonicity}). The \emph{monotone theory} aims to handle functions that are not monotone through the conjunction of $b$-monotone formulas.
\Cref{sec:monox} considers the (over)approximation of a formula by a $b$-monotone formula, the ``monotonization'' of a formula; \Cref{sec:monotone-hull} studies the conjunction of several such monotonizations through the \emph{monotone hull} operator. %

\subsection{Least $b$-Monotone Overapproximations}
\label{sec:monox}
\begin{definition}[$b$-Monotone Order]
\label{def:b-monotone-order}
Let $b$ be a cube. We define a partial order over states where $v \leq_b x$ when $x,v$ agree on all variables not present in $b$, and $x$ disagrees with $b$ on all variables on which also $v$ disagrees with $b$:
$\forall p \in \voc. \ x[p] \neq v[p] \mbox{ implies } p \in \cubdom{b} \land v[p]=b[p]$.
\end{definition}
Intuitively, $v \leq_b x$
when $x$ can be obtained from $v$ by flipping bits to the opposite of their value in $b$.
\begin{definition}[$b$-Monotonicity]
\label{def:b-monotonicity}
A formula $\psi$ is $b$-monotone for a cube $b$ if
$
\forall v \leq_b x. \ v \models \psi \mbox{ implies } x \models \psi.
$
\end{definition}
That is, if $v$ satisfies $\psi$, so do all the states that are farther away from $b$ than $v$.
For example, if $\psi$ is $000$-monotone and $100 \models \psi$, then because $100 \leq_{000} 111$ (starting in $100$ and moving away from $000$ can reach $111$), also $111 \models \psi$.
In contrast, $100 \not\leq_{000} 011$ (the same process cannot flip the $1$ bit that already disagrees with $000$), so $011$ does not necessarily belong to $\psi$.
($000$-monotonicity corresponds to the usual notion of monotone formulas.)

\begin{definition}[Least $b$-Monotone Overapproximation]
\label{def:monox}
For a formula $\varphi$ and a cube $b$, the \emph{least $b$-monotone overapproximation} of $\varphi$ is a formula $\monox{\varphi}{b}$ defined by
\begin{equation*}
	x \models \monox{\varphi}{b} \mbox{ iff } \exists v. \ v \leq_b x 	\land	 v \models \varphi.
\end{equation*}
\end{definition}
For example, if $100 \models \varphi$, then $100 \models \monox{\varphi}{000}$ because $\monox{\varphi}{000}$ is an overapproximation, and hence $111 \models \monox{\varphi}{000}$ because it is $000$-monotone, as above.
Here, thanks to minimality, $011$ does not belong to $\monox{\varphi}{000}$, unless $000$, $001$, $010$, or $011$ belong to $\varphi$.

The minimality property of $\monox{\varphi}{b}$ is formalized as follows:
\begin{lemma}
\label{lem:monox-minimality}
$\monox{\varphi}{b}$ (\Cref{def:monox}) is the least $b$-monotone formula $\psi$ (\Cref{def:b-monotonicity}) s.t.\ $\varphi \implies \psi$
(i.e., for every other $b$-monotone formula $\psi$, if $\varphi \implies \psi$ then $\monox{\varphi}{b} \implies \psi$).
\end{lemma}

An immediate but useful fact is that $\monox{\cdot}{b}$ is a monotone operator:
\begin{lemma}
\label{lem:bshouty-monox-monotone}
If $\varphi_1 \implies \varphi_2$ then $\monox{\varphi_1}{b} \implies \monox{\varphi_2}{b}$.
\end{lemma}

\para{Syntactic intuition}
Ordinary monotone formulas are $\vec{0}$-monotone;
for general $b$, a formula $\psi$ in DNF is $b$-monotone if interchanging $p,\neg p$ whenever $b[p]=\true$ results in a formula that is monotone DNF per the standard definition.\footnote{
	When $b$ is a full cube, another way to say this is that $\psi$ is $b$-monotone if it is monotone in the ordinary sense under the translation~\cite{wiedemann1987hamming} specified by $b$.
}
When $\varphi$ is not $b$-monotone, the monotonization $\monox{\varphi}{b}$ is the ``closest thing'', in the sense that it is the smallest $b$-monotone $\psi$ s.t.\ $\varphi \implies \psi$. As we shall see, $\monox{\varphi}{b}$ can be efficiently obtained from $\varphi$ by deleting literals.
The syntactic viewpoint is key for our results in~\Cref{sec:efficient-monox} and~\Cref{sec:cdnf}.

\para{Geometric intuition}
Geometrically, $\psi$ is $b$-monotone if $v \models \psi \implies x \models \psi$ for every states $v,x$ s.t.\ $v \leq_b x$; the partial order $\leq_b$ indicates that $x$ is ``farther away'' from $b$ in the Hamming cube than $v$ from $b$, namely, that there is a shortest path w.r.t.\ Hamming distance from %
$b$ to $x$
(or from %
$\project{x}{b}$---%
the projection of $x$ onto $b$---to $x$, when $b$ is not a full cube) that goes through $v$. %
A formula $\psi$ is $b$-monotone when it is closed under this operation, of getting farther from $b$.
In this way, $\monox{\varphi}{b}$ corresponds to the set of states $x$ %
to which there is a shortest path from $b$
that intersects $\varphi$.\footnote{
	This is reminiscent of visibility in Euclidean geometry~\cite[e.g.][]{DBLP:reference/cg/ORourke04}: picturing $b$ as a guard, the source of visibility, then $\monox{\varphi}{b}$ is the set of states that are visible in $\neg \varphi$, that is, the set of states $\sigma$ s.t.\ the ``line segment'' $[b,\sigma]$ is contained in $\neg \varphi$. Here $[b,\sigma]$ is the Hamming interval~\cite[e.g.][]{wiedemann1987hamming} between $b,\sigma$, the union of all the multiple shortest paths between the states (each path corresponds to a different permutation of the variables on which the states disagree).
}
The geometric viewpoint is key for our results in~\Cref{sec:efficient-monox} and for the abstract domain in~\Cref{sec:efficient-eedpr}.

\para{Disjunctive form}
The monotone overapproximation can be obtained from a DNF representation of the original formula, a fact that is useful for algorithms that compute the monotone overapproximation.
Starting with a DNF representation of $\varphi$, we can derive a DNF representation of $\monox{\varphi}{b}$
by dropping in each term the literals that agree with $b$.
Intuitively, if $\ell$ agrees with $b$, the ``constraint'' that $\sigma \models \ell$ is dropped from $\monox{t}{b}$ because
if $\sigma \models \monox{t}{b}$ then flipping the value of $\ell$ in $\sigma$ results in a state $\tilde{\sigma}$ such that $\sigma \leq_{b} \tilde{\sigma}$ and hence also $\tilde{\sigma} \models \monox{t}{b}$.

\begin{lemma}
\label{lem:bshouty-mon-mindnf}
Let $\varphi = t_1 \lor \ldots \lor t_m$ in DNF. Then the monotonization $\monox{\varphi}{b} \equiv \monox{t_1}{b} \lor \ldots \lor \monox{t_m}{b}$ where $\monox{t_i}{b} \equiv t_i \setminus b = \bigwedge \set{\ell \in t_i \land \ell \not\in b}$.
\end{lemma}
This fact has several useful corollaries.
First, for the important special case of a state (full cube) $v$, the monotonization $\monox{v}{b}$ is the conjunction of all literals that hold in $v$ except those that are present in $b$, written %
$$
	\cubemon{v}{b} \eqdef \monox{v}{b} =
					 \bigwedge{\set{p \ | \ v[p]=\true, \, p \not\in b}} \land
					 \bigwedge{\set{\neg p_i \ | \ v[p]=\false, \, \neg p \not\in b}}.
$$
In particular, if $v \models \varphi$ then $\cubemon{v}{b} \implies \monox{\varphi}{b}$ (follows from~\Cref{lem:bshouty-mon-mindnf} thinking about the representation $v \lor \varphi$).
A similar property holds under the weaker premise that $v$ is known to belong to the monotonization:
\begin{lemma}
\label{lem:moncube-model-of-monotonization}
If $v \models \monox{\varphi}{b}$ then $\cubemon{v}{b} \implies \monox{\varphi}{b}$.
\end{lemma}
Another corollary is that the DNF size cannot increase from $\varphi$ to $\monox{\varphi}{b}$:
\begin{lemma}
\label{lem:bshouty-monox-dnfsize}
$\dnfsize{\monox{\varphi}{b}} \leq \dnfsize{\varphi}$.
\end{lemma}

\subsection{Monotone Hull}
\label{sec:monotone-hull}
We now define the monotone hull, which is a conjunction of $b$-monotone overapproximations over all the $b$'s from a fixed set of states $B$.
\begin{definition}[Monotone Hull]
\label{def:monotone-hull}
The \emph{monotone hull} of a formula $\varphi$ w.r.t.\ a set of states $B$ is $\mhull{\varphi}{B} = \bigwedge_{b \in B}{\monox{\varphi}{b}}$. %
\end{definition}

The monotone hull can be simplified to use a succinct DNF representation of the basis $B$
instead of a conjunction over all states.
\begin{lemma}
\label{lem:mhull-dnf-base}
If $B \equiv b_1 \lor \ldots \lor b_m$ where $b_1,\ldots,b_m$ are cubes, then $\mhull{\varphi}{B} \equiv \monox{\varphi}{b_1} \land \ldots \land \monox{\varphi}{b_m}$.
\end{lemma}
Note that when $B = b$ is a single cube, $\mhull{\varphi}{b} = \monox{\varphi}{b}$.

Similarly to $\monox{\varphi}{b}$, the monotone hull is an overapproxmation:
\begin{lemma}
\label{lem:mhull-overapproximation}
$\varphi \implies \mhull{\varphi}{B}$.
\end{lemma}
\label{sec:monotone-basis}
In general, $\mhull{\varphi}{B}$ is not equivalent to $\varphi$. However, we can always choose $B$ so that $\mhull{\varphi}{B} \equiv \varphi$. A set $B$ that suffices for this is called a basis:
\begin{definition}[Monotone Basis]
\label{def:monotone-basis}
A \emph{monotone basis} is a set of states $B$.
It is a basis \emph{for a formula} $\varphi$ if
$
	\varphi \equiv \mhull{\varphi}{B}.
$
\end{definition}

Conversely, given a set $B$, we are interested in the set of formulas for which $B$ forms a basis:
\begin{definition}[Monotone Span]
\label{def:monotone-span}
$\mspan{B} = \set{\mhull{\varphi}{B} \mid \varphi \mbox{ over } \voc}$,
the set of formulas for which $B$ is a monotone basis.
\end{definition}
The following theorem provides a syntactic characterization of $\mspan{B}$, as the set of all formulas that can be written in CNF using clauses that exclude states from the basis.
The connection between CNF and the monotone basis is useful in~\Cref{sec:cdnf} where a monotone basis is constructed automatically, and in~\Cref{sec:efficient-eedpr} where it is used to define an abstract domain.
\begin{theorem}[\cite{DBLP:journals/iandc/Bshouty95}]
\label{lem:basis-conj-monox}
$\varphi \in \mspan{B}$ iff
there exist clauses $c_1,\ldots,c_s$ such that $\varphi \equiv c_1 \land \ldots \land c_s$ and for every $1 \leq i \leq s$ there exists $b_j \in B$ such that $b_j \not \models c_i$.
\end{theorem}
In particular, a basis $B$ for $\varphi$
can be constructed by writing a CNF representation of $\varphi$ and choosing for $B$ a state $b_j \not\models c_j$ for each clause $c_j$.

\para{Exact learning using the monotone theory}
\label{sec:monotone-theory-learning-background}
The monotone theory was first developed by Bshouty for the purpose of exact learning formulas that are not monotone. Essentially, the idea is to reconstruct the formula $\varphi$ by finding a monotone basis $B = \set{b_1,\ldots,b_t}$ for it, and constructing $\mhull{\varphi}{B}$ while using equivalence and membership queries. The CDNF algorithm~\cite{DBLP:journals/iandc/Bshouty95} achieves efficient learning in terms of the DNF \& CNF size of the target formula, and its code is shown in~\refappendix{sec:bshouty-cdnf}. Our CDNF invariant inference algorithm (\Cref{sec:cdnf}) is inspired by it, although it departs from it in significant ways (see~\Cref{rem:comparison-bshouty-cdnf}).  
\section{Super-Efficient Monotonization}
\label{sec:efficient-monox}
In this section we develop an efficient procedure to compute $\monox{\varphi}{b}$, which is a technical enabler of the results in following sections.
The algorithm, presented in~\Cref{alg:efficient-monox}, satisfies the following:
\begin{theorem}
\label{thm:efficient-monox-efficiency}
Let $\varphi$ be a formula and $\bkcube$ a cube.
The algorithm $\textsc{\superefficientalgname}(\varphi,\bkcube)$ computes $\monox{\varphi}{\bkcube}$ in $\bigO(n^2 \dnfsize{\monox{\varphi}{\bkcube}})$ SAT queries and time.
\end{theorem}
(Throughout this paper, $n$ denotes the number of propositional variables $\card{\voc}$.)

What distinguishes~\Cref{thm:efficient-monox-efficiency} is that the complexity bound depends on the DNF size of the \emph{output}, the monotonization $\monox{\varphi}{\bkcube}$, and \emph{not} on the size of the \emph{input} $\varphi$, in contrast to the algorithm by Bshouty~\cite{DBLP:journals/iandc/Bshouty95} (see~\Cref{remark:bshouty-monotonization}).

\begin{algorithm}[H]
\caption{Super-Efficient Monotonization}
\label{alg:efficient-monox}
\vspace{-0.5cm}
\begin{multicols}{2}
\begin{algorithmic}[1]
\begin{footnotesize}
\Procedure{\superefficientalgname}{$\varphi$, $\bkcube$}
	\State $H$ $\gets$ $\false$
	\While{$\SAT{\varphi \land \neg H}$} \label{ln:lambda-dual-itp:inv-monotonization-loop}
		\State \textbf{let} $\sigma_r \models \varphi \land \neg H$ $\label{ln:lambda-dual-itp:positive-example}$
		\State $v$ $\gets$ \Call{generalize}{$\varphi$, $\bkcube$, $\sigma_r$} $\label{ln:efficient-monox:call-gen}$
		\State $H$ $\gets$ $H \lor \cubemon{v}{\bkcube}$ $\label{ln:efficient-monox:disjoin-cube}$
	\EndWhile
	\State \Return $H$
\EndProcedure
\\
\\
\\
\\
\Procedure{generalize}{$\varphi$, $\bkcube$, $\sigma_r$}
	\State $v \gets \sigma_r$; walked $\gets$ $\true$
	\While{walked} $\label{ln:lambda-dual-itp:gen-post-loop}$
		\State walked $\gets$ $\false$
		\For{$j=1,\ldots,n$} \label{ln:lambda-dual-itp:check-closer-loop}
			\If{$\bkcube[p_j] = v[p_j]$}
				\State \textbf{continue}
			\EndIf
			\State $x \gets v[p_j \mapsto \bkcube[p_j]]$ $\label{ln:lambda-dual-itp:check-closer}$
			\If{$\SAT{\varphi \land \hamminginterval{x}{\project{x}{\bkcube}}}$} $\label{ln:efficient-monox-cube-disjointness}$
				\State $v \gets x$; walked $\gets$ $\true$
			\EndIf
		\EndFor
	\EndWhile
	\State \Return $v$
\EndProcedure
\end{footnotesize}
\end{algorithmic}
\end{multicols}
\vspace{-0.3cm}
\end{algorithm}
Starting from the candidate $H = \false$, the algorithm iteratively samples---through satisfying models of a SAT query---states that belong in $\varphi$ but not yet included in $H$. Every such state $\sigma_r$ generates a new term in $H$. Since $H$ is supposed to be $\bkcube$-monotone, the minimal term to include is $\moncube{\sigma_r}{\bkcube}$. To be efficient, the algorithm generalizes each example, trying to flip bits to find an example $v$ that also should be included in $\monox{\varphi}{\bkcube}$ and is closer in Hamming distance to $\bkcube$, which would result in a smaller term $\moncube{v}{\bkcube}$, thereby including more states in each iteration and converging faster. The criterion for $v$ is that a bit cannot be flipped if this would result in a state $x$ where the Hamming interval $\hamminginterval{x}{\project{x}{\bkcube}}$ does not intersect $\varphi$. Here, $\project{x}{\bkcube}$,
the projection~\cite[e.g.][]{wiedemann1987hamming} of $x$ onto the (possibly partial) cube $\bkcube$ is the state s.t.\
\begin{equation*}
\project{x}{\bkcube} = \begin{cases}
		\bkcube[p] & p \in \cubedom{\bkcube}
		\\
		x[p] & \text{otherwise}
	\end{cases},
 \end{equation*}
and the Hamming interval $\hamminginterval{\sigma_1}{\sigma_2}$ between two states $\sigma_1,\sigma_2$ is the smallest cube that contains both---the conjunction of the literals where these agree.
In sum, $\hamminginterval{x}{\project{x}{\bkcube}}$ is the conjunction of the literals where $x,\bkcube$ agree and the literals of $x$ over variables that are not present in $\bkcube$.
As we will show, $\hamminginterval{x}{\project{x}{\bkcube}}$ intersecting with $\varphi$ is an indicator for $x$ belonging to the monotonization of $\varphi$. %

The use of SAT queries in the algorithm does not necessarily assume that $\varphi$ is given explicitly, and indeed in~\Cref{sec:cdnf} we apply this algorithm with an implicit representation of $\varphi$ (using additional copies of the vocabulary).

The rest of this section proves~\Cref{thm:efficient-monox-efficiency}.
First, the result $v$ of generalization is so that when we disjoin the term $\moncube{v}{\bkcube}$, we do not ``overshoot'' to include states that do not belong to the true monotonization: %
\begin{lemma}
\label{lem:lambda-itp-gen-in-inv}
\label{lem:sem-underapprox}
If $\sigma_r \models \varphi$,
then $\textsc{generalize}(\varphi, \bkcube, \sigma_r)$ returns $v$ s.t.\ $\moncube{v}{\bkcube} \implies \monox{\varphi}{\bkcube}$. %
\end{lemma}
\begin{proof}
$v$ is chosen s.t.\ $\hamminginterval{v}{\project{v}{\bkcube}} \cap \varphi \neq \emptyset$---note that this holds trivially in the initial choice of $v$ which is $\sigma_r \models \varphi$.
Let $\tilde{\sigma} \models \varphi$ s.t.\ $\tilde{\sigma} \models \hamminginterval{v}{\project{v}{\bkcube}}$.
The latter means that $\tilde{\sigma} \leq_{\bkcube} v$, because $\hamminginterval{v}{\project{v}{\bkcube}}$
consists of all the literals in $v$ except for those that disagree with $\bkcube$, so $\sigma$ agrees with $v$ whenever $v,\bkcube$ agree. %
In more detail, $\tilde{\sigma}$ agrees with $v$ on all $p \not\in \cubedom{\bkcube}$ (because $\project{v}{\bkcube}[p] = v[p]$ on such variables), and for $p \in \cubedom{p}$, if $\tilde{\sigma}[p] \neq v[p]$, if $v,\bkcube$ agree on $p$ then likewise $\tilde{\sigma}$ agrees with them (because then $v[p]=\project{v}{\bkcube}[p]$ and $p$ is retained in the conjunction that forms the Hamming interval), which satisfies~\Cref{def:b-monotonicity}.
As also $\tilde{\sigma} \models \varphi$, this implies that $v \models \monox{\varphi}{\bkcube}$ per~\Cref{def:monox}.
Hence $\moncube{v}{\bkcube} \implies \monox{\varphi}{\bkcube}$, by~\Cref{lem:moncube-model-of-monotonization}.
\qed
\end{proof}
This shows that it is reasonable to disjoin the term $\moncube{v}{\bkcube}$ to $H$ in the hope of eventually obtaining $H = \monox{\varphi}{\bkcube}$.
The following lemma argues that the algorithm continues to sample states until
it converges to the true monotonization.
\begin{lemma}
\label{lem:lambda-itp-monox}
$\textsc{\superefficientalgname}(\varphi, \bkcube)$ terminates and returns $\monox{\varphi}{\bkcube}$.
\end{lemma}
\begin{proof}
First we show that when it terminates, the result is correct.
Always $H \subseteq \monox{\varphi}{\bkcube}$, because in each iteration we disjoin to $H$ a formula that satisfies the same property, by~\Cref{lem:sem-underapprox}.
The algorithm terminates when $\varphi \subseteq H$, and $H$ is always a $\bkcube$-monotone formula (%
by~\cref{lem:bshouty-mon-mindnf} the monotonization of $H$ is $H$ itself, which is $b$-monotone as in~\Cref{lem:monox-minimality}). From the minimality of $\monox{\varphi}{\bkcube}$ (%
\Cref{lem:monox-minimality}), necessarily also $\monox{\varphi}{\bkcube} \subseteq H$.

To show termination it suffices to show that $H$ strictly increases in each iteration (because the number of non-equivalent propositional formulas is finite).
To see this, note that \textsc{generalize}($\varphi$, $\bkcube$, $\sigma_r$) returns $v$ s.t.\ $v \leq_{\bkcube} \sigma_r$, since the procedure starts with $\sigma_r$ and only flips literals to agree with $\bkcube$. %
This implies that $\sigma_r \models \cubemon{v}{\bkcube}$, so after the iteration $\sigma_r \models H$ whereas previously $\sigma_r \not\models H$.
\qed
\end{proof}

The novelty of the algorithm is its efficiency, which we now turn to establish.
The crucial point is the generalization is able to produce, term by term, a minimal representation
of $\monox{\varphi}{\bkcube}$. To this end, we first show that $\cubemon{v}{\bkcube}$ that the algorithm computes in~\crefrange{ln:efficient-monox:call-gen}{ln:efficient-monox:disjoin-cube} is a \emph{prime implicant} of $\monox{\varphi}{\bkcube}$.
Recall that a term $t$ is an \emph{implicant} of a formula $\psi$ if $t \implies \psi$, and it is \emph{prime} if this no longer holds after dropping a literal, that is, for every $\ell \in t$ (as a set of literals), $\left(\land \left(t \setminus \set{\ell}\right)\right) \not\implies \psi$. It is \emph{non-trivial} if $t \not\equiv \false$ (not an empty set of literals).
\begin{lemma}
\label{lem:lambda-itp-gen-closest}
If $\sigma_r \models \varphi$, then $\textsc{generalize}(\varphi, \bkcube, \sigma_r)$ returns $v$ s.t.\ $\cubemon{v}{\bkcube}$ is a non-trivial prime implicant of $\monox{\varphi}{\bkcube}$.
\end{lemma}
\begin{proof}
\Cref{lem:sem-underapprox} shows that it is an implicant. %
It is non-trivial because $\sigma_r$ is a model of it, as shown as part of the proof of~\Cref{lem:lambda-itp-monox}.
Suppose that $\moncube{v}{\bkcube}$ is not prime.
Then there a literal over some variable $p$ that can be dropped. It is present in $\moncube{v}{\bkcube}$, which means that $p \in \cubedom{\bkcube}$ and $v[p] \neq \bkcube[p]$. Then the cube obtained from dropping the literal can be written as $\moncube{x}{\bkcube}$ where $x = v[p \mapsto \neg v[p]]$. If this cube is an implicant of $\monox{\varphi}{\bkcube}$, then, because $x \models \moncube{x}{\bkcube}$, in particular $x \models \monox{\varphi}{\bkcube}$. By~\Cref{def:monox}, there is $\tilde{\sigma} \models \varphi$ such that $\tilde{\sigma} \leq_{\bkcube} x$.
But the latter implies that $\tilde{\sigma} \in \hamminginterval{x}{\project{x}{\bkcube}}$, because, by~\Cref{def:b-monotonicity}, for every $p \not\in \cubedom{\bkcube}$, $\tilde{\sigma}[p] = x[p] = \project{x}{\bkcube}[p]$ and for every $p \in \cubedom{\bkcube}$ where $x,\project{x}{\bkcube}$ agree also $\tilde{\sigma},\bkcube$ agree (because $\project{x}{\bkcube}[p] = \bkcube[p]$).
Thus $\hamminginterval{x}{\project{x}{\bkcube}} \cap \varphi \neq \emptyset$, in contradiction to the choice of $v$, according to the check in~\cref{ln:efficient-monox-cube-disjointness}.
\qed
\end{proof}
Thanks to the fact that $\monox{\varphi}{\bkcube}$ is $\bkcube$-monotone, through one of the basic properties of monotone functions that dates back to Quine~\cite{quine1954two}, a prime implicant reproduces a term from the %
(unique) minimal representation of $\monox{\varphi}{\bkcube}$. We use this to show that the monotonization is computed in few iterations:
\begin{lemma}
\label{lem:lambda-itp-monox-efficiency}
The number of iterations of the loop in~\cref{ln:lambda-dual-itp:inv-monotonization-loop} of $\textsc{\superefficientalgname}(\varphi, \bkcube)$ is at most $\dnfsize{\monox{\varphi}{\bkcube}}$.
\end{lemma}
\begin{proof}
\Cref{lem:lambda-itp-gen-closest} shows that in each iteration we disjoin a prime implicant.
It is a property of monotone functions that they have a unique DNF representations with irredundant, which consists of the disjunction of all non-trivial prime implicants~\cite{quine1954two}, and this extends to $b$-monotone functions (through a simple renaming of variables to make the function monotone). Thus the non-trivial prime implicant we disjoin is a term of the minimal DNF representation of $\monox{\varphi}{\bkcube}$.
Each additional $\sigma_r$ produces a new term, as shown in~\Cref{lem:lambda-itp-monox}. %
\qed
\end{proof}

We are now ready to prove that the algorithm overall is efficient.
\begin{proof}[Proof of~\Cref{thm:efficient-monox-efficiency}]
By~\Cref{lem:lambda-itp-monox-efficiency} the number of iterations of the loop in~\cref{ln:lambda-dual-itp:inv-monotonization-loop} is bounded by $\dnfsize{\monox{\varphi}{\bkcube}}$. Each iteration calls \textsc{generalize}, which performs at most $n$ iterations of the loop in~\cref{ln:lambda-dual-itp:gen-post-loop} because the same variable is never flipped twice. Each iteration of this loop performs $n$ SAT queries in~\cref{ln:efficient-monox-cube-disjointness}. Note that the cube $\hamminginterval{x}{\project{x}{\bkcube}}$ is straightforward to compute in linear time. %
\qed
\end{proof}

\begin{remark}
\label{remark:bshouty-monotonization}
Bshouty~\cite{DBLP:journals/iandc/Bshouty95} used an algorithm for computing $\monox{\varphi}{\bkcube}$ whose complexity is bounded by the DNF input size $\dnfsize{\varphi}$, whereas \Cref{alg:efficient-monox}'s complexity is bounded by the DNF \emph{output} size, $\dnfsize{\monox{\varphi}{\bkcube}}$, which is never worse (\Cref{lem:bshouty-monox-dnfsize}), and sometimes significantly smaller. When considered as learning algorithms, the improved complexity of~\Cref{alg:efficient-monox} comes at the expense of the need for richer queries: Bshouty's algorithm is similar to~\Cref{alg:efficient-monox} (using an equivalence query in~\cref{ln:lambda-dual-itp:inv-monotonization-loop} that produces a positive example---see \textsc{genMQ} of~\Cref{alg:cdnf-exact} in~\Cref{sec:bshouty-cdnf}), except that the condition in~\cref{ln:efficient-monox-cube-disjointness} is replaced by checking whether $x \models \varphi$. This is a membership query to $\varphi$, whereas our check amounts to a disjointness query~\cite{DBLP:journals/ml/Angluin87}.
\end{remark}  
\section{Efficient Inference of CDNF Invariants}
\label{sec:cdnf}

\begin{figure}
\noindent
\begin{minipage}[t]{0.51\textwidth}
\vspace{-0.5cm}
\begin{algorithm}[H]
\caption{Dual of model-based interpolation-based inference~\cite{DBLP:conf/hvc/ChocklerIM12,DBLP:conf/lpar/BjornerGKL13}}
\label{alg:dual-itp-termmin}
\begin{algorithmic}[1]
\begin{footnotesize}
\Procedure{\dualourinterpolationalgname}{$\Init$,$\tr$,$\Bad$,$s$}
	\If{$\bmc{\tr}{\Init}{s} \cap \Bad \neq \emptyset$}
		\State \textbf{unsafe} $\label{ln:dual-itp-unsafe}$
	\EndIf
	\State $\varphi \gets \neg \Bad$ $\label{ln:dual-itp-frame0}$
    \While{$\varphi$ not inductive}
				\State \textbf{let} $\sigma, \sigma' \models \varphi \land \tr \land \neg \varphi'$ $\label{ln:dual-itp-cex}$
		\If{$\bmc{\tr}{\Init}{s} \cap \set{\sigma} \neq \emptyset$} $\label{ln:dual-itp-restart}$
			\State \textbf{restart} with larger $s$
		\EndIf
		\State take minimal clause $c \subseteq \neg \sigma$ s.t.\ $\bmc{\tr}{\Init}{s} \implies c$  $\label{ln:dual-itp-bmc}$
		\State $\varphi \gets \varphi \land c$ $\label{ln:dual-itp-learn}$
	\EndWhile
	\State \Return $\varphi$
\EndProcedure
\end{footnotesize}
\end{algorithmic}
\end{algorithm}
 \end{minipage}%
\hspace{2pt}\begin{minipage}[t]{0.49\textwidth}
\vspace{-0.5cm}
\begin{algorithm}[H]
\caption{Interpolation-based inference of CDNF invariants\vspace{2pt}}
\label{alg:lambda-itp}
\begin{algorithmic}[1]
\begin{footnotesize}
\Procedure{CDNF-ITP}{$\Init$,$\tr$,$\Bad$,$s$}
	\If{$\bmc{\tr}{\Init}{s} \cap \Bad \neq \emptyset$}
		\State \textbf{unsafe} $\label{ln:lambda-itp-unsafe}$
	\EndIf
	\State $\varphi \gets \neg \Bad$ $\label{ln:lambda-dual-itp-frame0}$
	\While{$\varphi$ not inductive} $\label{ln:lambda-dual-itp:ind-check}$
		\State \textbf{let} $\sigma, \sigma' \models \varphi \land \tr \land \neg \varphi'$ $\label{ln:lambda-dual-itp:cti}$
		\If{$\bmc{\tr}{\Init}{s} \cap \set{\sigma} \neq \emptyset$} $\label{ln:lambda-dual-itp:check-no-overapprox}$
			\State \textbf{restart} with larger $s$ $\label{ln:lambda-dual-itp-termmin:fail}$
		\EndIf
		\State $H$ $\gets$ \Call{\superefficientalgname}{$\bmc{\tr}{\Init}{s}$, $\sigma$} $\label{ln:lambda-dual-itp:block}$
		\State $\varphi \gets \varphi \land H$ $\label{ln:lambda-dual-itp:add-conjunct}$
	\EndWhile
	\State \Return $I$
\EndProcedure
\end{footnotesize}
\end{algorithmic}
\end{algorithm}  \end{minipage}
\end{figure}

In this section we build on the algorithm of~\Cref{sec:efficient-monox} to devise a new model-based interpolation-based algorithm that can efficiently infer invariants that have poly-size CNF and DNF representations (dubbed ``CDNF invariants'').
We start with background on the theoretical condition that guarantees the success of the original model-based algorithm for simpler forms of invariants.

\subsection{Background: Interpolation With the Fence Condition}
\label{sec:background-itp}
The essence of interpolation-based invariant inference (ITP) is to generalize a proof of \emph{bounded} unreachability---i.e., bounded model checking~\cite{DBLP:conf/tacas/BiereCCZ99}---into a proof of \emph{un}bounded reachability, that is, a part of the inductive invariant.
The reader is more likely to be familiar with the structure of the original algorithm by McMillan~\cite{DBLP:conf/cav/McMillan03}, which uses bounded unreachability \emph{to} the \emph{bad} states, where in each iteration the algorithm adds states to the candidate by \emph{disjoining} a formula from which it is impossible to reach a bad state in $s$ steps. %
However, for our purposes here, it is more convenient to consider the \emph{dual} version,
which uses bounded unreachability \emph{from} the \emph{initial} states, where in each iteration the algorithm excludes states from the candidate by \emph{conjoining} a formula which does not exclude any state that the system can reach in $s$ steps from an initial state (i.e., the candidate $\varphi$ is updated by $\varphi \gets \varphi \land H$ where $H$ is a formula s.t.\ $\bmc{\tr}{\Init}{s} \subseteq H$).

The original interpolation-based algorithm by McMillan uses a procedure that relies on the internals of the SAT solver~\cite{DBLP:conf/cav/McMillan03}. Complexity bounds on interpolation-based algorithms analyze later approaches that exercise control on how interpolants are generated, and do this in a model-based fashion~\cite{DBLP:conf/hvc/ChocklerIM12,DBLP:conf/lpar/BjornerGKL13} inspired by IC3/PDR~\cite{ic3,pdr}.
The dual %
version
is presented in~\Cref{alg:dual-itp-termmin}.
After starting the candidate $\varphi$ as $\neg \Bad$, each iteration checks for a counterexample to induction (\cref{ln:dual-itp-cex}), whose pre-state $\sigma$ is excluded from $\varphi$ at the end of the iteration (\cref{ln:dual-itp-learn}). Many states are excluded in each iteration beyond the counterexample, by conjoining to the candidate a minimal \emph{clause} that excludes $\sigma$ but retains all the states that are reachable in the system in $s$ steps (\cref{ln:dual-itp-bmc}---this involves up to $n$ queries of $s$-BMC, each time dropping a literal and checking whether the clause is still valid). If the counterexample cannot be blocked, because it is in fact reachable in $s$ steps,
this is an indication that $s$ needs to be larger (\cref{ln:dual-itp-restart}) to find a proof or a safety violation. 
The algorithm detects that the transition system is unsafe in~\cref{ln:dual-itp-unsafe} when $s$ is enough to find an execution from $\Init$ to $\Bad$ with at most $s$ transitions. (Our analysis of the algorithms in the paper focuses on the safe case, the complexity of finding an invariant.)

A condition that guarantees that $s$ is large enough for~\Cref{alg:dual-itp-termmin} to successfully find an inductive invariant, called the \emph{fence condition}, was recently put forward~\cite{DBLP:journals/pacmpl/FeldmanSSW21}, involving the Hamming-geometric boundary of the invariant.
\begin{definition}[Boundary]
Let $I$ be a set of states. Then the (inner) boundary of $I$, denoted $\boundarypos{I}$, is the set of states $\sigma^+ \models I$ s.t.\ there is a state $\sigma^-$ that differs from $\sigma$ in exactly one variable, and $\sigma^+ \models I, \sigma^- \not\models I$.
\end{definition}
\begin{definition}[Fence Condition]
Let $I$ be an inductive invariant for a transition system $(\Init,\tr,\Bad)$ and $s \in \mathbb{N}$. Then $I$ is $s$-forwards fenced if $\boundarypos{I} \subseteq \bmc{\tr}{\Init}{s}$.
\end{definition}
\begin{example}
\label{ex:fence-twobits}
Let $I$ be the set of all states where at least two bits are $0$ \emph{and} at least two bits are $1$. %
Then $\boundarypos{I}$ is the set where exactly two bits are $0$ (and at least two bits are $1$) or exactly two bits are $1$ (and at least two bits are $0$). %
Note that $I \setminus \boundarypos{I}$ contains many (most) states---those where three or more bits are $0$ and three or more bits are $1$.
The fence condition requires only from the states in $\boundarypos{I}$ to be reachable in $s$ steps.
\TODO{make this into a transition system}
\end{example}
The fence condition guarantees that throughout the algorithm's execution, $\varphi$ contains the fenced invariant, giving rise to the following correctness property:
\begin{theorem}[\cite{DBLP:journals/pacmpl/FeldmanSSW21}]
If there exists an $s$-forwards fenced invariant for $(\Init,\tr,\Bad,s)$, then $\mbox{\dualourinterpolationalgname}(\Init,\tr,\Bad,s)$ successfully finds an invariant.
\end{theorem}
The idea is that the property that $I \implies \varphi$ is maintained because the fence condition ensures that it suffices to verify that the clause that \textsc{generalize} computes, starting from a state $\sigma \models \varphi$ that is outside of $I$, does not exclude a state from $\bmc{\tr}{\Init}{s}$ to guarantee that it also does not exclude a state from $I$.
(Note that the fence condition does not provide a way to know whether an \emph{arbitrary} state belongs to $I$.)

The fence condition ensures the algorithm's success, but not that it is efficient---the number of iterations until convergence may be large, even when there is a fenced inductive invariant that has a short representation in CNF. (Note that in~\Cref{alg:dual-itp-termmin}, $\varphi$ is always in CNF.)
This was ameliorated in~\cite{DBLP:journals/pacmpl/FeldmanSSW21} by the assumption that the invariant is \emph{monotone}:
\begin{theorem}[\cite{DBLP:journals/pacmpl/FeldmanSSW21}]
If $I$ is an $s$-forwards fenced inductive invariant for $(\Init,\tr,\Bad,s)$ and $I$ is \emph{monotone} (can be written in CNF/DNF with all variables un-negated), then $\mbox{\dualourinterpolationalgname}(\Init,\tr,\Bad,s)$ successfully finds an inductive invariant in $\bigO(\cnfsize{I})$ inductiveness checks, and $\bigO\left(n \cdot \cnfsize{I}\right)$ checks of $s$-BMC and time.\footnote{
	A similar result applies when the invariant is \emph{unate}, that is, can be written in CNF/DNF so that every variable is either always negated or always un-negated.
}
\end{theorem}
However, when $I$ is \emph{not} monotone, it is possible for the algorithm to require an exponential number of iterations even though $\cnfsize{I}$ is small (and in fact, even though \emph{every} representation of $I$ without redundant clauses is small).

The challenge that we address in this section is to create an invariant inference algorithm that efficiently infers inductive invariants that are not monotone, while relying only on the fence condition.

\subsection{CDNF Inference With the Fence Condition}
We now present our new invariant inference algorithm (\Cref{alg:lambda-itp}), that is guaranteed to run in time polynomial in $n,\cnfsize{I},\dnfsize{I}$ of a fenced invariant $I$:
\begin{theorem}
\label{thm:lambda-dual-itp-efficient}
Let $I$ be a forwards $s$-fenced inductive invariant for $(\Init,\tr,\Bad)$.
Then $\textsc{CDNF-ITP}(\Init,\tr,\Bad,s)$ finds an inductive invariant in at most $\cnfsize{I} \cdot \dnfsize{I} \cdot n^2$ of $s$-BMC checks, $\cnfsize{I}$ inductiveness checks, and $\bigO(\cnfsize{I} \cdot \dnfsize{I} \cdot n^2)$ time.
\end{theorem}
\begin{example}
$I$ from~\Cref{ex:fence-twobits} has poly-size representations in both CNF and DNF. We write:
\begin{align*}
	\intertext{\emph{DNF}: there is a choice of four bits with two $0$'s and two $1$'s,}
		I \quad &\equiv \quad \bigvee_{1 \leq i_1 \neq i_2 \neq i_3 \neq i_4 \leq n}{\left(x_{i_1} = 0 \land x_{i_2} = 0 \land x_{i_3} = 1 \land x_{i_4} = 1\right)}.
	\\
	\intertext{\emph{CNF}: it is impossible that $n-1$ bits or more are $1$, likewise for $0$,}
		I \quad &\equiv \quad \left(\bigwedge_{i=1}^{n}{\bigvee_{j \neq i} x_j=0}\right) \land \left(\bigwedge_{i=1}^{n}{\bigvee_{j \neq i} x_j=1}\right).
\end{align*}
The CNF formula has $2n$ clauses, and the DNF has $\binom{n}{4} = \Theta(n^4)$ terms.
\Cref{thm:lambda-dual-itp-efficient} shows that such $I$ satisfying the fence condition can be inferred in a number of queries and time that is polynomial in $n$.
Note that these formulas fall outside the classes that previous results can handle efficiently (see~\Cref{sec:related}) as they are not monotone nor almost-monotone (the number of terms/clauses with negated variables is not constant). %
\end{example}

As noted by Bshouty~\cite{DBLP:journals/iandc/Bshouty95}, the class of formulas with short DNF and CNF includes the formulas that can be expressed by a small \emph{decision tree}:
a binary tree in which every internal node is labeled by a variable and a leaf by $\true$/$\false$, and $\sigma$ satisfies the formula if the path defined by starting from the root, turning left when the $\sigma$ assigns $\false$ to the variable labeling the node and right otherwise, reaches a leaf $\true$. The size of a decision tree is the number of leaves in the tree.
We conclude that (see the proof in~\refappendix{sec:proofs-appendix}):
\begin{corollary}
\label{cor:decision-tree-efficient}
Let $I$ be a forwards $s$-fenced inductive invariant for $(\Init,\tr,\Bad)$, that can be expressed as a decision tree of size $m$.
Then $\textsc{CDNF-ITP}(\Init,\tr,\Bad,s)$ finds an inductive invariant in at most $m^2 \cdot n^2$ of $s$-BMC checks, $m$ inductiveness checks, and $\bigO(m^2 \cdot n^2)$ time.
\end{corollary}

The algorithm \textsc{CDNF-ITP} which attains~\Cref{thm:lambda-dual-itp-efficient} is presented in~\Cref{alg:lambda-itp}. Its overall structure is similar to~\Cref{alg:dual-itp-termmin},
except the formula used to block a counterexample is the monotonization of the $s$-reachable states.
Specifically, starting from the candidate $\varphi = \true$ (\cref{ln:lambda-dual-itp-frame0}, the algorithm iteratively samples counterexamples to induction (\cref{ln:lambda-dual-itp:cti}) and blocks the pre-state $\sigma$ from $\varphi$ by conjoining $\monox{\bmc{\tr}{\Init}{s}}{\sigma}$, computed by invoking~\Cref{alg:efficient-monox}.
The SAT queries of the form $\SAT{\varphi \land \theta}$ that~\Cref{alg:efficient-monox} performs (see~\Cref{sec:efficient-monox}) have $\varphi = \bmc{\tr}{\Init}{s}$, and they amount to the BMC checks of whether $\bmc{\tr}{\Init}{s} \cap \theta \overset{?}{=} \emptyset$.

It is important for the efficiency result that \Cref{alg:lambda-itp} uses~\Cref{alg:efficient-monox} as a subprocedure. Using Bshouty's procedure (see~\Cref{remark:bshouty-monotonization}) would yield a bound of $n \cdot \dnfsize{\bmc{\tr}{\Init}{s}}$ checks of $s$-BMC, and it is likely that $\bmc{\tr}{\Init}{s}$ is complex to capture in a formula when $s$ is significant (as common for sets defined by exact reachability, such as the set of the reachable states).

We now proceed to prove the correctness and efficiency of the algorithm (\Cref{thm:lambda-dual-itp-efficient}).
Throughout, assume that $I$ is an inductive invariant for $(\Init,\tr,\Bad)$. $I$ will be $s$-forwards fenced; we state this explicitly in the premise of lemmas where this assumption is used.
The idea behind the correctness and efficiency of~\Cref{alg:lambda-itp} is that $\monox{I}{\sigma}$ is a stronger
formula than the clauses that are produced in~\Cref{alg:dual-itp-termmin}, causing the candidate to converge down to the invariant in fewer iterations, while never excluding states that belong to $I$ (because $I \subseteq \monox{I}{\sigma}$, as used in~\Cref{lem:lambda-itp-candidate}). As we will show (in~\Cref{lem:lambda-itp-overapprox}), this strategy results in a number of iterations that is bounded by the CNF size of $I$ (without further assumptions on the syntactic structure of $I$).
The trick, however, is to show (in~\Cref{lem:monox-bmc-inv}) that what the algorithm computes in~\cref{ln:lambda-dual-itp:block} is indeed $\monox{I}{\sigma}$, even though $I$ is unknown.
The crucial observation is that under the fence condition, the monotonization of the $s$-reachable states matches the monotonization of the invariant (even though these are different sets!). Note that this holds for any invariant that satisfies the fence condition.
To prove this we need to recall a fact about the monotonization of the boundary of a set:
\begin{lemma}[\cite{DBLP:journals/pacmpl/FeldmanSSW22}]
\label{lem:mhull-boundary}
Let $I,S$ be sets of states s.t.\ $\boundarypos{I} \subseteq S$ and $\sigma$ a state s.t.\ $\sigma \not\models I$.
Then $I \subseteq \monox{S}{\sigma}$.
\end{lemma}
The idea is that for every $x \in I$, there is a state on the boundary $v \in \boundarypos{I}$ s.t.\ $v \leq_\sigma x$ (where a shortest path between $x,b$ crosses $I$), and because also $v \in S$ we would have that $x \in \monox{S}{\sigma}$.

We proceed to relate the monotonizations of $\bmc{\tr}{\Init}{s},I$:
\begin{lemma}
\label{lem:monox-bmc-inv}
If $I$ is forwards $s$-fenced for $(\Init,\tr,\Bad)$, and $\sigma \not\models I$,
then $\monox{\bmc{\tr}{\Init}{s}}{\sigma} = \monox{I}{\sigma}$.
\end{lemma}
\begin{proof}
Since $I$ is an inductive invariant, $\bmc{\tr}{\Init}{s} \subseteq I$, so $\monox{\bmc{\tr}{\Init}{s}}{\sigma} \subseteq \monox{I}{s}$ from~\Cref{lem:bshouty-monox-monotone}.
For the other direction we use~\Cref{lem:mhull-boundary}: by the fence condition, $\boundarypos{I} \subseteq \bmc{\tr}{\Init}{s}$ and hence, as $\sigma \not\models I$, we obtain $I \subseteq \monox{\bmc{\tr}{\Init}{s}}{\sigma}$. By~\Cref{lem:monox-minimality}, this implies that $\monox{I}{\sigma} \subseteq \monox{\bmc{\tr}{\Init}{s}}{\sigma}$.
\qed
\end{proof}
We use this to characterize the candidate invariant the algorithm constructs:
\begin{lemma}
\label{lem:lambda-itp-candidate}
If $I$ is forwards $s$-fenced for $(\Init,\tr,\Bad)$,
then in each step of $\mbox{\textsc{CDNF-ITP}}(\Init,\tr,\Bad,k)$, $\varphi = \mhull{I}{\mathcal{C}_i} \land \neg \Bad$, where $\mathcal{C}_i$ is the set of counterexamples $\sigma$ the algorithm has observed so far.
In particular, $I \subseteq \varphi$.
\end{lemma}
\begin{proof}
First, $I \subseteq \varphi$ holds from the rest of the lemma because $I \subseteq \neg \Bad$ (it is an inductive invariant), and $I \subseteq \mhull{I}{\mathcal{C}_i}$ %
by~\Cref{lem:mhull-overapproximation}.
The proof of $\varphi = \mhull{I}{\mathcal{C}_i} \land \neg \Bad$ is by induction on iterations of the loop in~\cref{ln:lambda-dual-itp:ind-check}. Initially, $\mathcal{C} = \emptyset$ and indeed $\varphi = \neg \Bad$.
In each iteration, $I \subseteq \varphi$ using the argument above and the induction hypothesis.
Hence, the counterexample to induction of~\cref{ln:lambda-dual-itp:cti} has $\sigma \not\models I$ (otherwise $\sigma' \models I$ because $I$ is an inductive invariant, and this would imply also $\sigma \models \varphi$, in contradiction). %
Then \Cref{lem:lambda-itp-monox} ensures that $H = \monox{\bmc{\tr}{\Init}{s}}{\sigma}$.
\Cref{lem:monox-bmc-inv} shows that this is $\monox{I}{\sigma}$, as required.
\qed
\end{proof}
Essentially, the algorithm gradually learns a monotone basis (\Cref{def:monotone-basis}) for $I$ from the counterexamples to induction, and constructs $I$ via the monotone hull w.r.t.\ this basis. %
The next lemma shows that the size of the basis that the algorithm finds is bounded by $\cnfsize{I}$.
\begin{lemma}
\label{lem:lambda-itp-overapprox}
If $I$ is forwards $s$-fenced for $(\Init,\tr,\Bad)$,
then $\mbox{\textsc{CDNF-ITP}}(\Init,\tr,\Bad,k)$ successfully finds an inductive invariant.
Further, the number of iterations of the loop in~\cref{ln:lambda-dual-itp:ind-check} is at most $\cnfsize{I}$.
\end{lemma}
\begin{proof}
Since $\sigma \not\models I$, also it is not a model of the monotonization w.r.t.\ to itself, $\sigma \not\models \monox{I}{\sigma}$ (because the only state $x \leq_\sigma \sigma$ is $x=\sigma$---see~\Cref{def:b-monotone-order,def:monox}).
This shows, using~\Cref{lem:lambda-itp-candidate}, that at least one state is excluded from the candidate $\varphi$ in each iteration.
By the same lemma always $I \implies \varphi$, and the algorithm terminates when $\varphi$ is inductive, so this shows that the algorithm successfully converges to an inductive invariant.

To see that this occurs in at most $\cnfsize{I}$ iterations,
consider a minimal CNF representation of $I$, $I = c_1 \land \ldots \land c_{\cnfsize{I}}$. We argue that in each iteration produces at least one new clause from that representation, in the sense that for some $i$, $\varphi \land \monox{I}{\sigma_b} \implies c_i$ whereas previously $\varphi \notimplies c_i$.
Let $c_i$ be the clause that $\sigma \not\models c_i$ (recall e.g.\ that $\sigma \not\models \varphi$ and $I \subseteq \varphi$).
Then $\monox{I}{\sigma} \subseteq c_i$, since $c_i$ is $\sigma$-monotone
($\monox{c_i}{\sigma} = c_i$, using~\Cref{lem:bshouty-mon-mindnf}, because all the literals disagree with $\sigma$) and $I \subseteq c_i$, and $\monox{I}{\sigma}$ is the smallest such (\Cref{lem:monox-minimality}).
Thus when we conjoin $H = \monox{I}{\sigma}$ to $\varphi$ we conjoin at least one new $c_i$ that was not present in a CNF representation of $\varphi$; this can happen at most $\cnfsize{I}$ times.
\qed
\end{proof}

Overall:
\begin{proof}[Proof of~\Cref{thm:lambda-dual-itp-efficient}]
The algorithm's success in finding an invariant is established in~\Cref{lem:lambda-itp-overapprox}.
As for efficiency,
by~\Cref{lem:lambda-itp-overapprox}, there are at most $\cnfsize{I}$ iterations of the loop in \textsc{CDNF-ITP}, each performs a single inductiveness query, and calls \textsc{\superefficientalgname}. By~\Cref{thm:efficient-monox-efficiency} each such call performs at most $\bigO(n^2 \dnfsize{\monox{I}{\sigma}})$ $s$-BMC queries.
The claim follows because $\dnfsize{\monox{I}{\sigma}} \leq \dnfsize{I}$ (\Cref{lem:bshouty-monox-dnfsize}).
\qed
\end{proof}

\begin{remark}[Backwards fence condition]
\label{rem:backwards-fence}
Our main theorem in this section, \Cref{thm:lambda-dual-itp-efficient}, also has a dual version that applies to a fence condition concerning \emph{backwards} reachability. $I$ is $s$-backwards fenced if every state in the \emph{outer boundary} $\boundaryneg{I} \eqdef \boundarypos{\neg I}$ can reach a state in $\Bad$ in at most $s$ steps~\cite{DBLP:journals/pacmpl/FeldmanSSW21}. %
The dual of~\Cref{thm:lambda-dual-itp-efficient} is that there is an algorithm that achieves the same complexity bound under the assumption that $I$ is s-\emph{backwards} fenced (instead of $s$-forwards fenced). The dual-CDNF algorithm is obtained by running our CDNF algorithm on the dual transition system $(\Bad,\tr^{-1},\Init)$ (see e.g.~\cite[][Appendix A]{DBLP:journals/pacmpl/FeldmanISS20}) and negate the invariant; notice that the CDNF class is closed under negation.
This algorithm also achieves the same bound for decision trees as in~\Cref{cor:decision-tree-efficient}, under the backwards fence assumption.
\end{remark}

\begin{remark}[Comparison to Bshouty's CDNF algorithm]
\label{rem:comparison-bshouty-cdnf}
Our CDNF algorithm, \Cref{alg:lambda-itp}, is inspired by Bshouty's CDNF algorithm~\cite{DBLP:journals/iandc/Bshouty95}, but diverges from it in several ways.
The reason is the different queries available in each setting.
(The code for Bshouty's CDNF algorithm is provided in~\refappendix{sec:bshouty-cdnf}.)
Structurally, while the candidate in both algorithms is gradually constructed to be $\mhull{I}{\mathcal{C}_i} = \bigwedge_{\sigma \in \mathcal{C}_i}{\monox{I}{\sigma}}$ ($I$ being the unknown invariant/formula, and $\mathcal{C}_i$ the set of negative examples so far), \Cref{alg:lambda-itp} constructs each monotonization separately, one by one, whereas Bshouty's algorithm increases all monotonizations simultaneously.
Bshouty's design follows from having the source of examples---both positive and negative---equivalence queries, checking whether the candidate matches $I$. A membership query is necessary to decide whether the differentiating example is positive or negative for $I$ in order to decide whether to add disjuncts to the existing monotonizations or to add a new monotonization, respectively.
This procedure is problematic in invariant inference, because we cannot in general decide, for a counterexample $(\sigma,\sigma')$ showing that our candidate is not inductive, whether $\sigma \not\models I$ (negative) or $\sigma' \models I$ (positive)~\cite{ICELearning,DBLP:journals/pacmpl/FeldmanISS20}.
The solution in previous work~\cite{DBLP:journals/pacmpl/FeldmanSSW21} was to assume that the invariant satisfies both the forwards \emph{and} backwards fence condition (see~\Cref{rem:backwards-fence}). Under this assumption it is possible to decide whether $\sigma \models I$ for an arbitrary state $\sigma$. However, this condition is much stronger than a one-sided version of the fence condition.
Instead, in our inference algorithm, the candidate is ensured to be an overapproximation of the true $I$, so each counterexample to induction in~\cref{ln:lambda-dual-itp:cti} yields a negative example. Positive examples are obtained in~\cref{ln:lambda-dual-itp:positive-example} from $\bmc{\tr}{\Init}{s} \subseteq I$; there is no obvious counterpart to that in exact learning, because in that setting we have no a-priori knowledge of some set $S$ that underapproximates $I$, let alone one where we know---as the fence condition guarantees through~\Cref{lem:mhull-boundary}---that covering $S$ in the monotonization is enough to cover $I$.
\end{remark}  
\section{Efficient Implementation of Abstract Interpretation}
\label{sec:efficient-eedpr}
In this section we build on the algorithm of~\Cref{sec:efficient-monox} to prove a complexity upper bound on abstract interpretation in the domain based on the monotone theory (\Cref{thm:efficient-eepdr}). We begin with background on this domain.

\subsection{Background: Abstract Interpretation in the Monotone Theory}
Recall that given a set of states $B$, the monotone span (\Cref{def:monotone-span}) of $B$, $\mspan{B}$, is the set of formulas $\varphi$ s.t.\ $\mhull{\varphi}{B} \equiv \varphi$, or, equivalently, the set of formulas that can be written as conjunctions of clauses that exclude states from $B$ (\Cref{lem:basis-conj-monox}).
The \emph{abstract domain} $\madom{B}=\langle \mspan{B}, \implies, \join_{B}, \false \rangle$, introduced in~\cite{DBLP:journals/pacmpl/FeldmanSSW22}, is a join-semilattice over the monotone span of $B$, ordered by logical implication, with bottom element $\false$. %
The lub $\join_{B}$ exists because the domain is finite and closed under conjunction (follows from~\Cref{lem:basis-conj-monox}). A Galois connection $(2^{\States}, \subseteq) \galois{\malpha{B}}{\gamma} (\mspan{B}, \implies)$ with the concrete domain is obtained through the \emph{concretization} $\gamma(\varphi) = \set{\sigma \, | \, \sigma \models \varphi}$ and the \emph{abstraction} $\malpha{B}(\psi) = \mhull{\psi}{B}$~\cite{DBLP:journals/pacmpl/FeldmanSSW22}.\footnote{
$\madom{B}$ is parametrized by a choice of a monotone basis $B$. When $B$ is large, the abstraction is more precise; it is precise enough to prove safety when there exists an inductive invariant that can be expressed in CNF such that each clause excludes at least one state from $B$ (through~\Cref{lem:basis-conj-monox}). The fewer states $B$ includes, the more extrapolation is performed in each abstraction step. However, since $B$ also changes the available inductive invariants, the overall convergence might actually be faster with a larger (less extrapolating) $B$. Understanding how to choose $B$ is an important direction for future work. (In~\cite{DBLP:journals/pacmpl/FeldmanSSW22}, $B$ was obtained from the states that reach a bad state in a fixed number of steps, mimicking PDR’s scheme for generating proof obligations.)
}

Given a transition system $(\Init,\tr)$, iterations %
of abstract interpretation with the abstract transformer are given by
$\xi_0 = \malpha{B}(\Init), \xi_{i+1} = \malpha{B}\left(\tr(\gamma(\xi_i)) \cup \Init\right)$.
Substituting $\gamma,\malpha{B}$ %
yields the iterations as shown in~\Cref{alg:eepdr-ai}. %
\begin{algorithm}[H]
\caption{Kleene Iterations in $\madom{B}$}
\label{alg:eepdr-ai}
\begin{algorithmic}[1]
\begin{footnotesize}
\Procedure{AI-$\madom{B}$}{$\Init$, $\tr$, $\Bad$}
	\State $i \gets 0$
	\State $\Frameai_{-1} \gets \false$
	\State $\Frameai_0 \gets \mhull{\Init}{B}$ $\label{ln:eepdr-ai:frame0}$
	\While{$\Frameai_{i} \notimplies \Frameai_{i-1}$}
		\State $\Frameai_{i+1} = \mhull{\tr(\Frameai_i) \cup \Init}{B}$ $\label{ln:eepdr-ai:transformer}$
		\State $i \gets i+1$
	\EndWhile
	\State \Return $\Frameai_i$
\EndProcedure
\end{footnotesize}
\end{algorithmic}
\end{algorithm}

Each iterate in~\Cref{alg:eepdr-ai} involves a monotone hull (\cref{ln:eepdr-ai:frame0,ln:eepdr-ai:transformer}), which is a conjunction of monotonizations. Using~\Cref{alg:efficient-monox} this can be computed efficiently. We follow on this idea to prove efficient complexity upper bounds on~\Cref{alg:eepdr-ai}.

\subsection{Complexity Upper Bound}
To obtain a complexity upper bound on~\Cref{alg:eepdr-ai} we need to bound %
the time needed to compute each $\Frameai_i$
as well as the number of $\Frameai_i$'s.
A bound for the latter is provided by~\cite{DBLP:journals/pacmpl/FeldmanSSW22}:
\begin{theorem}[\cite{DBLP:journals/pacmpl/FeldmanSSW22}]
\label{thm:abstract-hyperdiamter-bound}
Let $(\Init,\tr,\Bad)$ be a transition system.
Then $\textsc{AI-$\madom{B}$}(\Init,\tr,\Bad)$ converges in %
iteration number at most
\begin{equation*}
	\aibound \, \eqdef \, \prod_{i=1}^{m}{\left(
					\dnfsize{\monox{\tr}{\reflect{\cubejoin{B}}\land\bkcube'_i}}
					+
					\dnfsize{\monox{\Init}{\bkcube_i}}
					\right)
				},
\end{equation*}
where
$B$ can be written in DNF as $\bkcube_1 \lor \ldots \lor \bkcube_m$, %
the cube $\cubejoin{B}$ consists of the literals that appear in all $\bkcube_1,\ldots,\bkcube_m$ (i.e., $\cubejoin{B} = \bigcap_{i=1}^{m}{\bkcube_i}$ as sets of literals), and the reflection of a cube $d = \ell_1 \land \ldots \land \ell_r$ is $\reflect{d} = \neg \ell_1 \land \ldots \land \neg \ell_r$.
\end{theorem}
We fix a DNF representation of $B = \bkcube_1 \lor \ldots \lor \bkcube_m$.
For brevity, we use $\aibound$ to refer to the bound in~\Cref{thm:abstract-hyperdiamter-bound}.
When $\aibound$ is small, it reflects the benefit of using abstract interpretation in $\mspan{B}$ over exact reachability (even though $\aibound$ is not always a tight bound)~\cite{DBLP:journals/pacmpl/FeldmanSSW22}.
An example of $\aibound$ for a simple system appears in~\Cref{ex:ai-parity}.

In this section we prove that it is possible to implement~\Cref{alg:eepdr-ai} so that its overall complexity is polynomial in the same quantity $\aibound$, the number of variables $n$, and the number $m$ of terms in the representation of $B$:
\begin{theorem}
\label{thm:efficient-eepdr}
\Cref{alg:eepdr-ai} can be implemented to terminate in $\bigO(n^2 \aibound + (n+m)\aibound^2)$
SAT queries and time.
\yotam{actually $\bigO(n^2 \aibound + m\aibound^2)$ queries and $\bigO(n^2 \aibound + (n+m)\aibound^2)$ time}
\end{theorem}

\begin{example}
\label{ex:ai-parity}
Let $n$ be an odd number.
Consider a transition system over $\vec{x} = x_1,\ldots,x_n$, where $\Init$ is $\vec{x} = 00\ldots00$ and the transition relation chooses an even number of variables that are $0$ from the initial state
and turns them into $1$.
If we take $B$ to be the singleton set containing the state $\vec{x} = 11\ldots11$ (hence, $B$ is a cube and $m=1$), then~\Cref{lem:bshouty-mon-mindnf} yields that $\monox{\tr}{\vec{x}=00\ldots00\land\vec{x}'=11\ldots11} = \bigvee_{i=1}^{n}{(x'_i = 0)}$ (see~\Cref{sec:parity-example-monotonization} for details) so $\aibound = \dnfsize{\monox{\tr}{\vec{x}=00\ldots00\land\vec{x}'=11\ldots11}} = \bigO(n)$.
\Cref{thm:efficient-eepdr} shows that an implementation of abstract interpretation in $\madom{B}$ for this system terminates in $\bigO(n^3)$ SAT queries and time.
This is significant because a naive implementation of~\Cref{alg:eepdr-ai} would start, for the first iteration of~\cref{ln:eepdr-ai:transformer}, by computing the exact post-image $\tr(\Init)$; in our example this is the set of states where the parity of $\vec{x}$ is $0$, which cannot be represented in polynomial-size DNF nor CNF~\cite[e.g.][]{DBLP:books/daglib/0028067}. Our implementation is able to compute the abstraction of the post-image without constructing the post-image and avoids the blowup in complexity.
\cite{DBLP:journals/pacmpl/FeldmanSSW22} contains other examples with small $\aibound$.
\end{example}

\yotam{easy example with multiple cubes: direct product of parities over disjoint parts of the states, $\cubejoin{B}$ not a problem if the only initial state in $\tr$ is $\vec{x}=00\ldots00$.}

At this point, the direct approach to implement~\Cref{alg:eepdr-ai} is to perform $\mhull{\varphi}{B}$ in~\cref{ln:eepdr-ai:frame0,ln:eepdr-ai:transformer} through $\bigwedge_{j=1}^{m}{\mbox{\Call{\superefficientalgname}{$\varphi$, $\bkcube_j$}}}$, invoking~\Cref{alg:efficient-monox} on $\varphi$.
Indeed, this achieves a bound that is only slightly worse than~\Cref{thm:efficient-eepdr} (see~\Cref{rem:simple-efficient-eepdr-ai}). In what follows we provide an implementation that both explicates the connection to $\aibound$, and achieves exactly the bound of~\Cref{thm:efficient-eepdr}.

\begin{algorithm}[H]
\caption{Efficient Kleene Iterations in $\madom{B}$}
\label{alg:efficient-eepdr-ai}
\begin{algorithmic}[1]
\begin{footnotesize}
\Procedure{AI-$\madom{B}$}{$\Init$, $\tr$, $\Bad$}
	\State $i \gets 0$
	\State $\Frameai_{-1} \gets \false$
	\State $\Frameai_0 \gets \bigwedge_{j=1}^{m}{\mbox{\Call{\superefficientalgname}{$\Init$, $\bkcube_j$}}}$ $\label{ln:efficient-eepdr-ai:frame0}$
	\For{$j=1..m$}
		\State $\absr{\tr}_j \gets \textsc{\superefficientalgname}(\tr \lor \Init', \reflect{\cubejoin{B}} \land \bkcube'_j)$ $\label{ln:efficient-eepdr-ai:mon-tr}$
	\EndFor
	\While{$\Frameai_{i} \notimplies \Frameai_{i-1}$}
		\State $\Frameai_{i+1} = \bigvee{\set{\left(\restrict{t_1}{\voc'}\right) \land \ldots \land \left(\restrict{t_m}{\voc'}\right) \ \Big{|} \ t_j \mbox{ a term of } \absr{\tr}_j, \; \exists \sigma_j \in \xi_i. \, \sigma_j \models \left(\restrict{t_j}{\voc}\right)}}$  $\label{ln:efficient-eepdr-ai:transformer}$
		\State $i \gets i+1$
	\EndWhile
	\State \Return $\Frameai_i$
\EndProcedure
\end{footnotesize}
\end{algorithmic}
\end{algorithm}

Our implementation is displayed in~\Cref{alg:efficient-eepdr-ai}.
The first iterate is computed as described above by invoking~\Cref{alg:efficient-monox} on $\Init$ (\cref{ln:efficient-eepdr-ai:frame0}). %
The SAT queries performed by~\Cref{alg:efficient-monox} are in this case straightforward, with $\varphi = \Init$.

To compute the next iterates, we first compute monotnizations of the concrete transformer, $\tr \lor \Init'$ (\cref{ln:efficient-eepdr-ai:mon-tr}).
This is a two-vocabulary formula, and accordingly the monotonizations are w.r.t.\ two-vocabulary cubes.
The monotonizations are computed in DNF form and stored in $\absr{\tr}_j$.
The next iterate $\Frameai_{i+1}$ is formed from the $\absr{\tr}_j$'s by taking all the combinations of terms from $\absr{\tr}_1,\ldots,\absr{\tr}_m$ whose pre-state part is satisfied by at least one state in $\xi_i$, and forming the conjunction of the post-state parts: for a term $t = \ell_1 \lor \ldots \lor \ell_{i_1} \lor \ell'_{i_1+1} \lor \ldots \lor \ell'_{i_2}$ over $\voc \uplus \voc'$, the restriction $\restrict{t}{\voc} = \ell_1 \lor \ldots \lor \ell_{i_1}$ and $\restrict{t}{\voc'} = \ell'_{i_1+1} \lor \ldots \lor \ell'_{i_2}$.

The intuition is that in the original algorithm, given a set of states $\Frameai_i$, we find the set of states in $\Frameai_{i+1}$ by taking the result of the transformer $\tr \lor \Init'$ on the specific $\Frameai_i$, then, for the monotonization of the result, adding also the states that are required by the $\bkcube_j$-monotone order, and this we do for every disjunct $\bkcube_j$ in $B$, letting $\Frameai_{i+1}$ be the conjunction of the said monotonizations. Here, instead, we monotonize $\tr \lor \Init'$ itself w.r.t.\ every $b_j$, such that for every pre-state we have ready the monotonization of the corresponding post-state. We then form $\Frameai_{i+1}$ by picking and conjoining the right monotonizations---the ones whose pre-state is in the previous frame.
(The monotonization of the pre-state w.r.t.\ $\reflect{\cubejoin{B}}$ is useful for decreasing $\aibound$, and hence the obtained bound, without altering $\Frameai_{i+1}$; the latter stems from the fact that the input $\Frameai_i$ is also the result of such a procedure, so the presence of a pre-state in $\Frameai_i$ indicates the presence all the states in its monotone hull w.r.t.\ $\reflect{\cubejoin{B}}$ in $\Frameai_i$.)

The invocation of~\Cref{alg:efficient-monox} in~\cref{ln:efficient-eepdr-ai:mon-tr} is on a double-vocabulary formula; still, the SAT queries to be performed in the invocation of~\Cref{alg:efficient-monox} are simple SAT queries about two-vocabulary formulas (and a counterexample is a pair of states). %

It is important for the efficiency result that \Cref{alg:lambda-itp} uses~\Cref{alg:efficient-monox} as a subprocedure. Using Bshouty's procedure (see~\Cref{remark:bshouty-monotonization}) would yield a bound in terms of the DNF size of the original transition relation, which could be significantly larger, especially in cases where the abstract interpretation procedure terminates faster than exact forward reachability~\cite{DBLP:journals/pacmpl/FeldmanSSW22}.

The rest of this section proves that~\Cref{alg:efficient-eepdr-ai} realizes~\Cref{thm:efficient-eepdr}.
To show that it correctly implements~\Cref{alg:eepdr-ai}, we need the following fact about~\Cref{alg:eepdr-ai}:
\begin{lemma}[\cite{DBLP:journals/pacmpl/FeldmanSSW22}]
\label{lem:hypertr-step}
In~\Cref{alg:eepdr-ai}, $\sigma' \models \Frameai_{i+1}$ iff there exist $\sigma_1,\ldots,\sigma_m \models \Frameai_i$ s.t.\
\begin{equation}
\label{eq:hypertr-step}
	(\sigma_1,\sigma') \models \monox{\tr \lor \Init'}{\reflect{\cubejoin{B}}\land\bkcube'_1} \land \ldots \land (\sigma_m,\sigma') \models \monox{\tr \lor \Init'}{\reflect{\cubejoin{B}}\land\bkcube'_m}.
\end{equation}
\end{lemma}
We use this to show the correctness of~\Cref{alg:efficient-eepdr-ai}:
\begin{lemma}
\label{lem:efficient-eepdr-correct}
$\Frameai_i$ in~\Cref{alg:efficient-eepdr-ai} is logically equivalent to $\Frameai_i$ in~\Cref{alg:eepdr-ai}.
\end{lemma}
\begin{proof}
By induction over $i$.
The correctness of $\Frameai_0$ follows from the correctness of~\Cref{alg:efficient-monox} (\Cref{thm:efficient-monox-efficiency}).
For the same reasons, $\absr{\tr}_j$ of~\Cref{alg:efficient-eepdr-ai} is equivalent to $\monox{\tr \lor \Init'}{\reflect{\cubejoin{B}} \land \bkcube'_j}$.
Now for some DNF manipulation: for every $\sigma'$,
\begin{align*}
	&\exists \sigma_1,\ldots,\sigma_m.
	(\sigma_1,\sigma') \models \absr{\tr}_1 \land \ldots \land (\sigma_m,\sigma') \models \absr{\tr}_m
\\
\iff \quad
	&\exists \sigma_1,\ldots,\sigma_m. \ \exists t_1 \mbox{ term of } \absr{\tr}_1, \ldots, \exists t_m \mbox{ term of } \absr{\tr}_m. \
	(\sigma_1,\sigma') \models t_1 \land \ldots \land (\sigma_m,\sigma') \models t_m
\\
\iff \quad
	&\exists \sigma_1,\ldots,\sigma_m. \ \exists t_1 \mbox{ term of } \absr{\tr}_1, \ldots, \exists t_m \mbox{ term of } \absr{\tr}_m. \
\\
	& \qquad
	\sigma_1 \models \left(\restrict{t_1}{\voc}\right) \land \sigma' \models \left(\restrict{t_1}{\voc'}\right) \land \ldots \land \sigma_1 \models \left(\restrict{t_m}{\voc}\right) \land \sigma' \models \left(\restrict{t_m}{\voc'}\right)
\\
\iff \quad
	&\exists \sigma_1, \, \exists t_1 \mbox{ term of } \absr{\tr}_1. \ \sigma_1 \models \left(\restrict{t_m}{\voc}\right) \land \sigma' \models \left(\restrict{t_m}{\voc'}\right)
\\
	&
	\land \ldots \land
\\
	&
	\exists \sigma_m, \, \exists t_m \mbox{ term of } \absr{\tr}_m. \ \sigma_m \models \left(\restrict{t_m}{\voc}\right) \land \sigma' \models \left(\restrict{t_m}{\voc'}\right).
\end{align*}
Hence, $\exists \sigma_1,\ldots,\sigma_m \in \Frameai_i$ of~\Cref{alg:efficient-eepdr-ai} that with $\sigma'$ satisfy~\Cref{eq:hypertr-step} iff $\sigma' \in \Frameai_{i+1}$ of~\Cref{alg:efficient-eepdr-ai}. \Cref{lem:hypertr-step} and the induction hypothesis complete the proof.
\qed
\end{proof}

We can now proceed to prove the complexity bound for~\Cref{alg:efficient-eepdr-ai}.
\begin{lemma}
\label{lem:efficient-eepdr-ai-efficiency}
\Cref{alg:efficient-eepdr-ai} terminates in $\bigO(n^2 \aibound + (n+m)\aibound^2)$ SAT queries and time.
\end{lemma}
\begin{proof}
By~\Cref{thm:efficient-monox-efficiency}, each invocation of~\Cref{alg:efficient-monox} in~\cref{ln:efficient-eepdr-ai:frame0} takes $\bigO\left(n^2 \cdot \dnfsize{\monox{\Init}{\bkcube_i}}\right) = \bigO(n^2 \aibound)$ queries and time.
Similarly, each invocation in~\cref{ln:efficient-eepdr-ai:mon-tr} takes $\bigO\left((2n)^2 \cdot \dnfsize{\monox{\tr \lor \Init'}{\reflect{\cubejoin{B}}\land\bkcube'_i}}\right)$ queries and time.
This quantity is $\bigO(n^2 \aibound)$, because by~\Cref{lem:bshouty-mon-mindnf} $\dnfsize{\monox{\tr \lor \Init'}{\reflect{\cubejoin{B}}\land\bkcube'_i}} \leq \dnfsize{\monox{\tr}{\reflect{\cubejoin{B}}\land\bkcube'_i}} + \dnfsize{\monox{\Init'}{\reflect{\cubejoin{B}}\land\bkcube'_i}} = \dnfsize{\monox{\tr}{\reflect{\cubejoin{B}}\land\bkcube'_i}} + \dnfsize{\monox{\Init}{\bkcube_i}}$.
In each iteration,
the number of combinations of terms in~\cref{ln:efficient-eepdr-ai:transformer} is at most
$\prod_{i=1}^{m}{
	\dnfsize{\monox{\tr \lor \Init'}{\reflect{\cubejoin{B}}\land\bkcube'_i}}
}$.
For each of the $m$ terms in the combination, we split the term to $\voc,\voc'$ parts in time linear term size which is at most $n$, and perform a SAT check for whether the term intersects $\Frameai_i$.
Overall this step involves $\bigO\left(m \cdot \aibound\right)$ queries and $\bigO\left((n+m) \cdot \aibound\right)$ time.
This is the cost of each iteration; the number of iterations is bounded by $\aibound$ by~\Cref{thm:abstract-hyperdiamter-bound}.
The claim follows.
\qed
\end{proof}

The proof of~\Cref{thm:efficient-eepdr} follows from~\Cref{lem:efficient-eepdr-correct,lem:efficient-eepdr-ai-efficiency}.

\begin{remark}
\label{rem:simple-efficient-eepdr-ai}
\Cref{lem:efficient-eepdr-correct,lem:efficient-eepdr-ai-efficiency} have the consequence that in~\Cref{alg:eepdr-ai}, $\dnfsize{\Frameai_i} \leq \aibound$ (interestingly, this is true in particular for the resulting inductive invariant). This is a proof that the direct implementation of the monotone hull by $m$ calls to~\Cref{alg:efficient-monox} amounts to $\bigO(n^2 m \aibound)$ SAT queries in each iteration, and $\bigO(n^2 m \aibound^2)$ time thanks to~\Cref{thm:abstract-hyperdiamter-bound}. Though asymptotically inferior, this implementation approach may be more efficient than~\Cref{alg:efficient-eepdr-ai} when $\dnfsize{\Frameai_i} \ll \aibound$.
\end{remark}  
\section{Related Work}
\label{sec:related}

\para{Complexity bounds for invariant inference}
Conjunctive/disjunctive invariants can be inferred in a linear number of SAT calls~\cite{DBLP:conf/fm/FlanaganL01,DBLP:conf/cade/LahiriQ09}.
On the other hand, inferring CNF/DNF invariants for general transition systems is $\NP$-hard with access to a SAT solver~\cite{DBLP:conf/cade/LahiriQ09}, even when the invariants are restricted to monotone formulas~\cite{DBLP:journals/pacmpl/FeldmanISS20}.
Complexity results for model-based interpolation-based invariant inference (stemming from the analysis of the algorithm by~\cite{DBLP:conf/hvc/ChocklerIM12,DBLP:conf/lpar/BjornerGKL13}) were presented in~\cite{DBLP:journals/pacmpl/FeldmanSSW21}, based on the (backwards) fence condition, which tames reachability enough to efficiently infer monotone and almost-monotone DNF invariants, or monotone and almost-monotone CNF invariants under the (forwards) fence condition with the dual algorithms.
Our algorithm (\Cref{alg:efficient-eepdr-ai}) can achieve the same bounds if, similar to~\cite{DBLP:journals/pacmpl/FeldmanSSW21}, the monotone basis (set of counterexamples) is fixed in advance. (This alteration is needed because some short monotone DNF formulas have large CNF size~\cite{MILTERSEN2005325}.)
Our algorithm is more versatile because it applies to CDNF formulas which are not monotone or almost-monotone, and alleviates the need to know a monotone basis in advance.
A CDNF complexity bound similar to ours was obtained in~\cite{DBLP:journals/pacmpl/FeldmanSSW21} under the much stronger assumption that both the backwards and the forwards fence condition hold simultaneously (see~\Cref{rem:backwards-fence}).
Property-directed reachability algorithms~\cite{ic3,pdr} were shown efficient on several parametrized examples~\cite{DBLP:conf/mbmv/SeufertS17,DBLP:journals/pacmpl/FeldmanSSW22} and the very special case of maximal transition systems for monotone invariants~\cite{DBLP:journals/pacmpl/FeldmanISS20}. We show (\Cref{sec:efficient-eedpr}) that abstract interpretation in the monotone theory, studied under the name $\Lambda$-PDR~\cite{DBLP:journals/pacmpl/FeldmanSSW22} in relation to standard PDR, is efficient in broader circumstances, when the DNF size of certain monotonizations of the transition relation are small, using the same quantity $\aibound$ that in previous work~\cite{DBLP:journals/pacmpl/FeldmanSSW22} was established as an upper bound on the number of iterations (without an overall complexity result). %

\para{Monotone theory in invariant inference}
The monotone theory has been employed in previous works on invariant inference.
The aforementioned previous results on inference under the fence condition~\cite{DBLP:journals/pacmpl/FeldmanSSW21} also employ the monotone theory, and are based on one-to-one translations of Bshouty's algorithms to invariant inference, replacing equivalence queries by inductiveness checks and membership queries by bounded model checking. For Bshouty's CDNF algorithm this only works under the stronger two-sided fence condition, which is why our algorithm differs significantly (see~\Cref{rem:comparison-bshouty-cdnf}). The one-sided fence condition suffices for the translation of Bshouty's $\Lambda$-algorithm, which is suitable when a monotone basis is known in advance, e.g.\ for almost-monotone invariants, whereas the CDNF algorithm learns a monotone basis on-the-fly.
Another translation of Bshouty's algorithm to invariant inference is by Jung et al.~\cite{DBLP:journals/mscs/JungKDWY15}, who combine the CDNF algorithm of Bshouty~\cite{DBLP:journals/iandc/Bshouty95} with predicate abstraction and templates to infer quantified invariants. They overcome the problem of membership queries in the original algorithm (see~\Cref{rem:comparison-bshouty-cdnf}) heuristically, using under- and over-approximations and sometimes random guesses, which could lead to the need to restart.
The monotone theory is also used in a non-algorithmic way in~\cite{DBLP:journals/pacmpl/FeldmanSSW22} to analyze overapproximation in IC3/PDR through abstract interpretation in the monotone theory, to which we prove a complexity upper bound.

\para{Inferring decision tree invariants}
In machine learning, decision trees are a popular representation of hypotheses and target concepts.
Garg et al.~\cite{DBLP:conf/popl/0001NMR16} adapt an algorithm by Quinlan~\cite{DBLP:journals/ml/Quinlan86} to infer invariants (later, general Horn clauses~\cite{DBLP:journals/pacmpl/EzudheenND0M18}) in the form of decision trees over numerical and Boolean attributes, which is guaranteed to converge, but not necessarily efficiently overall (even though the algorithm efficiently generates the candidates, the number of candidates could be large).
Similarly to our algorithm, the translation of the CDNF algorithm to invariant inference in~\cite{DBLP:journals/pacmpl/FeldmanSSW21} is applicable also to Boolean decision trees, but, as previously mentioned, requires the stronger two-sided fence condition, whereas our result is the first to do so under the (one-sided) fence condition.

\para{Complexity of abstract interpretation}
The efficiency of the abstract transformers is crucial to the overall success of abstract interpretation, which is often at odds with the domain accuracy; a famous example is the octagon abstract domain~\cite{DBLP:journals/lisp/Mine06}, whose motivation is the prohibitive cost of the expressive polyhedra domain~\cite{DBLP:conf/popl/CousotH78}. We provide a way to compute abstract transformers in the monotone span domain that is efficient in terms of the DNF size of the result (see also~\Cref{rem:simple-efficient-eepdr-ai}).
The computation of the abstract transformer in~\Cref{alg:efficient-eepdr-ai} is inspired by works in symbolic abstraction~\cite{DBLP:conf/vmcai/RepsSY04,DBLP:journals/entcs/ThakurLLR15} about finding \emph{representations} of the best abstract transformer, rather than computing it anew per input~\cite{DBLP:conf/sas/ElderLSAR11,DBLP:conf/sas/ThakurER12,DBLP:conf/vmcai/RepsT16}. %
\section{Conclusion}
\label{sec:conclusion}
This work has accomplished invariant inference algorithms with efficient complexity guarantees in two settings---model-based interpolation and property-directed reachability---resolving open problems where the missing component (as it turns out) was a new way to compute monotone overapproximations.
A common theme is the use of rich syntactic forms of the formulas that the algorithms maintain: in our model-based interpolation algorithm, the candidate invariant is a conjunction of DNFs, even though the target invariant has both short CNF and DNF representations. In our efficient implementation of abstract interpretation, each iterate is again a conjunction of DNFs, although the natural definition (\Cref{lem:basis-conj-monox}, and as inspired by PDR) is with CNF formulas.
We hope that these ideas could inspire new interpolation-based and property-directed reachability algorithms that would benefit in practice from richer hypotheses and techniques from the monotone theory.  
\subsubsection*{Acknowledgement}
We thank the anonymous reviewers and Hila Peleg for insightful comments.
The research leading to these results has received funding from the
European Research Council under the European Union's Horizon 2020 research and
innovation programme (grant agreement No [759102-SVIS]).
This research was partially supported by the Israeli Science Foundation (ISF) grant No.\ 1810/18. %
\bibliographystyle{splncs04}
\bibliography{refs}

\appendix
\clearpage
\section{Missing Proofs}
\label{sec:proofs-appendix}
\begin{proof}[Proof of~\Cref{lem:monox-minimality}]
That $\monox{\varphi}{b}$ is $b$-monotone overapproximation of $\varphi$ is immediate from the definition. For minimality, let $\psi$ be a $b$-monotone formula s.t.\ $\varphi \implies \psi$, we need to show that $\monox{\varphi}{b} \implies \psi$. Let $x \models \monox{\varphi}{b}$. Then, by definition, there is $v \models \varphi$ s.t.\ $v \leq_b x$. By the assumption that $\varphi \implies \psi$ also $v \models \psi$, and, because $\psi$ is assumed to be $b$-monotone, it follows that $x \models \psi$. The claim follows.
\qed
\end{proof}

\begin{proof}[Proof of~\Cref{lem:moncube-model-of-monotonization}]
From the premise it follows that $\cubemon{v}{b} \implies \monox{\monox{\varphi}{b}}{b}$, but by~\Cref{lem:bshouty-mon-mindnf} and conversions to DNF, $\monox{\monox{\varphi}{b}}{b} \equiv \monox{\varphi}{b}$.
\qed
\end{proof}

\begin{proof}[Proof of~\Cref{cor:decision-tree-efficient}]
A decision tree of size $m$ has a DNF representation of $m$ terms: a disjunction of terms representing the paths that reach a leaf $\true$, each is the conjunction of the variables on the path with polarity according to left/right branch. Similarly, it has a CNF representation of $m$ clauses: a conjunction of clauses which are the negations of paths that reach a leaf $\false$. Now apply~\Cref{thm:lambda-dual-itp-efficient}.
\qed
\end{proof} %
\section{Bshouty's CDNF Algorithm}
\label{sec:bshouty-cdnf}

\begin{algorithm}[H]
\caption{CDNF-algorithm: exact learning CDNF formulas~\cite{DBLP:journals/iandc/Bshouty95}}
\label{alg:cdnf-exact}
\vspace{-0.4cm}
\begin{multicols}{2}
\begin{algorithmic}[1]
\begin{footnotesize}
\Procedure{CDNF-Algorithm}{}
	\State $t \gets 0$
	\While{$x$ $\gets$ $\equivalencequery{\bigwedge_{i=1}^{t}{H_i}}$ is not $\bot$} \label{ln:cdnf-exact:branch-eq}
		\If{$x \models \bigwedge_{i=1}^{t}{H_i}$}
			\State $H_{t+1} \gets \false$; $b_{t+1} \gets x$ \label{ln:cdnf-exact:negative}
			\State $t \gets t+1$
		\Else
			\State $\sigma' \gets x$
			\For{$i=1,\ldots,t$}
				\If{$\sigma' \not\models H_i$}
					\State $d$ $\gets$ \Call{genMQ}{$\sigma'$, $b_i$}
					\State $H_i \gets H_i \lor d$
				\EndIf
			\EndFor
		\EndIf
	\EndWhile
	\State \Return $\bigwedge_{i=1}^{t}{H_i}$
\EndProcedure
\Procedure{genMQ}{$\sigma$, $b$}
	\State $v \gets \sigma$
    \State walked $\gets$ $\true$
	\While{walked}
		\State walked $\gets$ $\false$
		\For{$j=1,\ldots,n$}
			\If{$b[p_j] = v[p_j]$}
				\State \textbf{continue}
			\EndIf
			\State $x \gets v[p_j \mapsto b[p_j]]$
			\If{$\membershipquery{x} = \true$} $\label{ln:bshouty-exact:check-membership}$
				\State $v \gets$ x
				\State walked $\gets$ $\true$
			\EndIf
		\EndFor
	\EndWhile
	\State \Return{$\cubemon{v}{b}$}
\EndProcedure
\end{footnotesize}
\end{algorithmic}
\end{multicols}
\vspace{-0.4cm}
\end{algorithm}
The original CDNF algorithm is presented in~\Cref{alg:cdnf-exact}, aiming to identify an unknown formula $\psi$.
$\equivalencequery{H}$ returns $\bot$ if $H \equiv \psi$ and a differentiating counterexample otherwise.
$\membershipquery{x}$ returns $\true$ iff $x \models \psi$.

\section{Calculation of $\aibound$ in~\Cref{ex:ai-parity}}
\label{sec:parity-example-monotonization}
In~\Cref{ex:ai-parity}, $\tr$ can be written in DNF as
\newcommand*{\bbox}{%
  \tcboxmath[colback=white, colframe=black, size=fbox, arc=0pt, boxrule=0.4pt]%
}
\newcommand*{\nbox}{%
  \tcboxmath[colback=white, colframe=white, size=fbox, arc=0pt, boxrule=0.4pt]%
}
\small
\begin{align*}
	\tr
		=
		  &\nbox{(\vec{x}=\litabs{00\ldots000} \land \vec{x}'=00\ldots000)}
		 \\
		 \hline
		  \lor&\nbox{(\vec{x}=\litabs{00\ldots000} \land \vec{x}'=00\ldots0\litabs{11})}
		  \\
		  \lor&\nbox{(\vec{x}=\litabs{00\ldots000} \land \vec{x}'=00\ldots\litabs{1}0\litabs{1})}
		  \\
		  \lor&\nbox{(\vec{x}=\litabs{00\ldots000} \land \vec{x}'=0\litabs{1}\ldots00\litabs{1})}
		  \\
		  \lor&\nbox{(\vec{x}=\litabs{00\ldots000} \land \vec{x}'=\litabs{1}0\ldots00\litabs{1})}
		  \\
		  \lor&\nbox{(\vec{x}=\litabs{00\ldots000} \land \vec{x}'=00\ldots\litabs{1}\litabs{1}0)}
		  \\
		  \lor&\nbox{(\vec{x}=\litabs{00\ldots000} \land \vec{x}'=0\litabs{1}\ldots0\litabs{1}0)}
		  \\
		  \lor&\nbox{(\vec{x}=\litabs{00\ldots000} \land \vec{x}'=\litabs{1}0\ldots0\litabs{1}0)}
		  \\
		  \lor&\nbox{(\vec{x}=\litabs{00\ldots000} \land \vec{x}'=0\litabs{1}\ldots\litabs{1}00)}
		  \\
		  \lor&\nbox{(\vec{x}=\litabs{00\ldots000} \land \vec{x}'=\litabs{1}0\ldots\litabs{1}00)}
		  \\
		  \lor&\nbox{(\vec{x}=\litabs{00\ldots000} \land \vec{x}'=\litabs{1}\litabs{1}\ldots000)}
		  \\
		  \hline
		  \lor&\ldots
		  \\
		  \hline
		  \lor&\bbox{(\vec{x}=\litabs{00\ldots000} \land \vec{x}'=\litabs{1}\litabs{1}\ldots\litabs{1}\litabs{1}0)}
		 \\
		 \lor&\bbox{(\vec{x}=\litabs{00\ldots000} \land \vec{x}'=\litabs{1}\litabs{1}\ldots\litabs{1}0\litabs{1})}
		 \\
		 \lor&\bbox{\ldots}
		 \\
		 \lor&\bbox{(\vec{x}=\litabs{00\ldots000} \land \vec{x}'=\litabs{1}\litabs{1}\ldots0\litabs{1}\litabs{1})}
		 \\
		 \lor&\bbox{(\vec{x}=\litabs{00\ldots000} \land \vec{x}'=\litabs{1}0\ldots\litabs{1}\litabs{1}\litabs{1})}
		 \\
		 \lor&\bbox{(\vec{x}=\litabs{00\ldots000} \land \vec{x}'=0\litabs{1}\ldots\litabs{1}\litabs{1}\litabs{1})}
		.
\end{align*}
\normalsize
The horizontal lines groups disjuncts according to the number of variables that are turned to $1$---zero, two, etc., up to $n-1$.
The literals that appear in color are those that are dropped when computing a DNF for $\monox{\tr}{\vec{x}=00\ldots000\land\vec{x}'=11\ldots111}$ according to~\Cref{lem:bshouty-mon-mindnf}.
In this way we see that the monotonization of the last group of terms, indicated by boxes, subsumes the monotonization of the other terms, yielding
\begin{equation*}
	\monox{\tr}{\vec{x}=00\ldots000\land\vec{x}'=11\ldots111} \equiv \bigvee_{i=1}^{n}{(x'_i = 0)},
\end{equation*}
and hence $\aibound = \dnfsize{\monox{\tr}{\vec{x}=00\ldots000\land\vec{x}'=11\ldots111}} \leq n$. %
\end{document}